\documentclass[11pt,letterpaper]{article}
\usepackage{amsthm}
\usepackage{amsmath}
\usepackage{amssymb}
\usepackage{graphicx}
\usepackage{nicefrac} 
\usepackage{tikz}
\usepackage{listings}
\lstdefinelanguage{pseudocode}{%
  morekeywords={skip,while,for,if,then,else,or,not,fi,procedure,%
    return,do,od,end,output,Protocol},%
  escapechar={*},%
  mathescape=true,%
  columns=flexible%
}
\usepackage[left=1.2in,top=1.2in,right=1.2in,bottom=1.2in]{geometry}

\theoremstyle{definition}
\newtheorem{definition}{Definition}
\theoremstyle{plain}
\newtheorem{theorem}[definition]{Theorem}

\newtheorem{lemma}[definition]{Lemma}
\DeclareMathOperator{\poly}{poly}
\DeclareMathOperator{\polylog}{polylog}
\DeclareMathOperator{\Ext}{Ext}
\DeclareMathOperator{\negl}{negl}
\newcommand{\E}{\textbf{E}}

\newcommand{\cG}{\mathcal{G}}

\newcommand{\cO}{\mathcal{O}}
\newcommand{\cW}{\mathcal{W}}
\newcommand{\cZ}{\mathcal{Z}}
\newcommand{\Stilde}{\widetilde{S}}

\DeclareMathOperator{\Gen}{Gen}
\DeclareMathOperator{\GenS}{GenS}
\DeclareMathOperator{\size}{size}

\newcommand{\iffull}{\iftrue}\newcommand{\ifshort}{\iffalse}

\renewcommand{\>}{\rangle}
\begin{document}
\iffull\title{General Hardness Amplification of Predicates and Puzzles}\fi
\ifshort\title{General Hardness Amplification of Predicates and Puzzles}\fi

\iffull
\author{Thomas Holenstein\thanks{Department of Computer Science, ETH
    Zurich, \texttt{thomas.holenstein@inf.ethz.ch}.  Work was done
    while the author was at Microsoft Research, Silicon Valley
    Campus.} \and Grant Schoenebeck\thanks{Department of Computer Science, Princeton University, Princeton NJ 08544, USA, \texttt{gschoene@princeton.edu}. Work was partially done while
    author was a summer intern at Microsoft Research Silicon Valley
    Campus and
    partially supported by a National Science Foundation Graduate
    Fellowship.}}
\fi
\ifshort\author{Anonymous Submission}
\fi
\maketitle
\begin{abstract}
  We give new proofs for the hardness amplification of efficiently
  samplable predicates and of weakly verifiable puzzles which generalize to new settings.  More
  concretely, in the first part of the paper, we give a new proof of
  Yao's XOR-Lemma that additionally applies to related theorems in the cryptographic
  setting.  Our proof seems simpler than previous ones, yet
  immediately generalizes to statements similar in spirit such as the
  extraction lemma used to obtain pseudo-random generators from
  one-way functions [H\aa{}stad, Impagliazzo, Levin, Luby, SIAM J. on
  Comp. 1999].

  In the second part of the paper, we give a new proof of hardness
  amplification for weakly verifiable puzzles, which is more general
  than previous ones in that it gives the right bound even for an arbitrary
  monotone function applied to the checking circuit of the underlying
  puzzle.

  Both our proofs are applicable in many settings of interactive
  cryptographic protocols because they satisfy a property that we call
  ``non-rewinding".  In particular, we show that any weak
  cryptographic protocol whose security is given by the
  unpredictability of single bits can be strengthened with a natural
  information theoretic protocol.   As an example, we show how these
   theorems solve the main open question from [Halevi and Rabin,
   TCC2008] concerning bit commitment.
   \ifshort
   
   \fi
\end{abstract}

\ifshort
\thispagestyle{empty}
\setcounter{page}{0}
\fi

\section{Introduction}
In this paper, we study two scenarios of hardness amplification.  In
the first scenario, one is given a predicate~$P(x)$, which is
somewhat hard to compute given~$x$.  More concretely: $\Pr[A(x) =
P(x)] \leq 1-\frac{\delta}{2}$ for any $A$ in some given complexity
class, where typically~$\delta$ is not too close to $1$ but at least
polynomially big (say, $\frac{1}{\poly(n)} < \delta <
1-\frac{1}{\poly(n)}$).  One then aims to find a predicate which is
even harder to compute.

In the second scenario, one is given a computational search problem,
specified by some relation~$R(x,y)$.  One then assumes that no
algorithm of a certain complexity satisfies $\Pr[(x,A(x)) \in R] >
1-\delta$, and again, is interested in finding relations which are
even harder to satisfy.  It is sometimes the case that~$R$ may only be
efficiently computable given some side information generated while
sampling~$x$.  Such problems are called ``weakly verifiable puzzles''.

Our aim is to give proofs for theorems in both scenarios which are
both simple and versatile.  In particular, we will see that our proofs
are applicable in the interactive setting, where they give
stronger results than those previously known.

\subsection{Predicates}
\paragraph{Overview and previous work}
Roughly speaking, Yao's XOR-Lemma \cite{Yao82} states that if a
predicate $P(x)$ is somewhat hard to compute, then the $k$-wise XOR
$P^{\oplus k}(x_1,\ldots,x_k) := P(x_1) \oplus \dots \oplus P(x_k)$
will be even harder to compute.  While intuitive, such statements are
often somewhat difficult to prove.  The first proof of the
above appears to be by Levin \cite{Levin87} (see also
\cite{GoNiWi95}).  In some cases, even stronger statements are needed:
for example, the extraction lemma states that one can even extract
several bits out of the concatenation $P(x_1)P(x_2)\dots P(x_k)$, which look pseudorandom to a distinguisher given $x_1,\ldots,x_k$.
Proving this statement for tight parameters is considered the
technically most difficult step in
the original proof that one-way functions imply pseudorandom
generators \cite{HILL99}.  Excluding this work, the
easiest proof available seems to be based on Impagliazzo's hard-core
set theorem \cite{Impagl95}, more concretely the uniform version of it
\cite{Holens05, BaHaKa09}.  A proof along those lines is given in
\cite{Holens06, HaHaRe06b}.  Similar considerations are true for the
more efficient proof that one-way functions imply pseudorandom
generators given by Haitner et al.\cite{HaReVa10}.

\paragraph{Contributions of this paper}

In this paper, we are concerned with statements of a similar nature as
(but which generalize beyond) Yao's XOR-Lemma.  We give a new theorem,
which is much easier to prove than the hard-core set theorem, and
which is still sufficient for all the aforementioned applications.

Our main observation can be described in relatively simple terms.  In
the known proof based on hard-core sets (\cite{Impagl95,Holens05}),
the essential statement is that there is a large set $S$, such that
for~$x\in S$ it is computationally difficult to predict $P(x)$ with a
non-negligible advantage over a random guess.  Proving the existence
of the set $S$ requires some work (basically, boosting, as shown in
\cite{KliSer99}).  We use the idea that the set $S$ can be
made \emph{dependent} on the circuit which attempts to predict $P$.
The existence of a hard set $S$ for a particular circuit is a much
easier fact to show (and occurs as a building block in some proofs of
the hard-core theorem).  For our idea to go through, $S$ has to be
made dependent on some of the inputs to~$C$ as well as some other
fixed choices.  This technique of switching quantifiers resembles
a statement in \cite{BaShWi03}, where Impagliazzo's hard-core
set theorem is used to show that in some
definitions of pseudo-entropy it is also possible to switch
quantifiers.

Besides being technically simpler, making the set~$S$ dependent on~$C$
has an additional advantage.  For example, consider a proof of the XOR
Lemma.  To get a contradiction, a circuit~$C$ is assumed which does
well in predicting the XOR, and a circuit~$D$ for a single instance is
built from~$C$.  On input~$x$, $D$ calls~$C$ as a subroutine several
times, each time ``hiding''~$x$ as one of the elements of the input.
Using our ideas, we can ensure that $x$ is hidden always in the same
place $i$, and even more, the values of the inputs
$x_1,\ldots,x_{i-1}$ are constant and independent of $x$.  This
property, which we call non-rewinding, is useful in the case one wants
to amplify the hardness of interactive protocols.

We remark that in this paper we are not concerned with
efficiency of XOR-Lemmas in the sense of derandomizing them (as in, e.g.,
\cite{ImpWig97, ImJaKa06, IJKW08}).

\subsection{Weakly Verifiable Puzzles}
\paragraph{Overview and Previous Work}
The notion of weakly verifiable puzzles was introduced by Canetti et
al.~\cite{CaHaSt05}.  A weakly verifiable puzzle consists of a
sampling method, which produces an instance~$x$ together with a
circuit~$\Gamma(y)$, checking solutions.  The task is, given~$x$ but
not necessarily~$\Gamma$, to find a string~$y$ for which $\Gamma(y) =
1$.  One-way functions are an example: $\Gamma(y)$ just outputs~$1$ if
$f(y) = x$ (since $\Gamma$ depends on the instance it can contain
$x$).  However, weakly verifiable puzzles are more general, since
$\Gamma$ is not given at the time $y$ has to be found.

Canetti et al.~show that if no efficient algorithm finds solutions
with probability higher than~$\delta$, then any efficient algorithm
finds $k$ solutions simultaneously with probability at
most~$\delta^k+\epsilon$, for some negligible $\epsilon$.  This result
was strengthened by \cite{ImJaKa09}, showing that requiring
some~$\delta' > \delta + 1/\poly(n)$ fraction of correct answers
already makes efficient algorithms fail, if~$k$ is large enough.
Independently of the current work, Jutla \cite{Jutla10} improved their
bound to make it match the standard Chernoff bound. A different
strengthening was given in \cite{HalRab08}, where it was noted that
the algorithm in \cite{CaHaSt05} has an additional property which
implies that it can be applied in an interactive cryptographic
setting, also they studied how much easier solving a weakly verifiable
puzzle becomes if one simply asks for a single correct solution from
$k$ given puzzles.  Also independently of our work, Chung et
al.~\cite{CLLY09} give a proof for the threshold case (similar to
Jutla) which is also applicable in an interactive setting; however,
their parameters are somewhat weaker than the ones given by most other
papers.  Finally, \cite{DIJK09} gives yet another
strengthening: they allow a weakly verifiable puzzle to have multiple
solutions indexed by some element $q$, and the adversary is allowed to
interactively obtain some of them.  They then study under what
conditions the hardness is amplified in this setting.

\paragraph{Contributions of this paper}
In this work, we present a theorem which unifies and strengthens the
results given in \cite{CaHaSt05,HalRab08,ImJaKa09,Jutla10,CLLY09}: assume a
monotone function $g: \{0,1\}^k \to \{0,1\}$ specifies which
subpuzzles need to be solved in order to solve the resulting puzzle
(i.e., if $c_1,\ldots,c_k$ are bits where $c_i$ indicates that a valid
solution for puzzle $i$ was found, then $g(c_1,\ldots,c_k)=1$ iff this
is sufficient to give a valid solution for the overall case.)  Our
theorem gives a tight bound for any such $g$ (in this sense, previous
papers considered only threshold functions for $g$).  Furthermore, as
we will see our proof is also applicable in an interactive setting
(the proofs given in \cite{ImJaKa09, Jutla10} do not have this
property).  Our proof is heavily inspired by the one given in
\cite{CaHaSt05}.

\subsection{Strengthening Cryptographic Protocols}
\paragraph{Overview and Previous Work}
Consider a cryptographic protocol, such as bit commitment.  Suppose
that a non-perfect implementation of such a protocol is given, which
we would like to improve.  For example, assume that a cheating receiver
can guess the bit committed to with some probability, say~$3/5$.
Furthermore, suppose that a cheating sender can open the commitment in
two ways with some probability, say~$1/5$.  Can we use this protocol
to get a stronger bit commitment protocol?

Such questions have been studied in various forms both in the
information theoretic and the computational model \cite{DaKiSa99,
  DFMS04, DwNaRe04, Holens05, HolRen05, Wullsc07, HalRab08}.

However, all of the previous computational work except \cite{HalRab08}
focused on the case where the parties participating in the protocol
are at least semi-honest, i.e., they follow the protocol correctly
(this is a natural assumption in the case for the work on key agreement
\cite{DwNaRe04, Holens05, HolRen05}, as in this case the participating
parties can be assumed to be honest).  An exception to this trend was
the work by Halevi and Rabin \cite{HalRab08}, where it was shown that
for \emph{some} protocols, the information theoretic bounds also apply
computationally.

The above are results in case where the protocol is repeated
\emph{sequentially}.  The case where the protocol is repeated in
parallel is more complicated
\cite{BeImNa97,PieWik07,PasVen07,HPWP10,Haitne09,ChuLiu10}.  


\paragraph{Contributions of this paper}
We explicitly define ``non-rewinding" (which was, however, pointed to
in \cite{HalRab08}) which helps to provide a sufficient condition for
transforming complexity theoretic results into results for
cryptographic protocols.  Using, the above results, and specifically
that the above results are non-rewindable, we show that we can
strengthen any protocol in which the security goal is to make a bit
one party has unpredictable to the other party, in the case where an
information theoretic analogue can be strengthened.  We also study
interactive weakly verifiable puzzles (as has been done implicitly in
\cite{HalRab08}), and show that natural ways to amplify the hardness
of these work.

We only remark that our proof is applicable to parallel repetition for
non-interactive (two-round) protocols (e.g. CAPTCHAs).

\ifshort Due to space restrictions, many of the proofs and even some
of the formal statements of theorems have been moved to the appendix.
\fi

\section{Preliminaries}

\begin{definition} Consider a circuit $C$ which has a tuple of
  designated input wires labeled $y_1,\ldots,y_{k}$.  An oracle
  circuit $D(\cdot)$ with calls to $C$ is \emph{non-rewinding} if
  there is a fixed $i$ and fixed strings $y_{1}^*$ to $y_{i-1}^*$ such
  that for any input $y$ to $D$, all calls to $C$ use inputs
  $(y_1^*,\ldots,y_{i-1}^*,y)$ on the wires labeled $y_1,\ldots,y_i$.
\end{definition}

\begin{definition}
  Let~$C$ be a circuit which has a block of input wires labeled $x$.
  An oracle circuit~$D$ which calls $C$ (possibly several times)
  treats $x$ obliviously if the input $x$ to $D$ is forwarded to $C$
  directly, and not used in any other way in $D$.
\end{definition}

We say that an event happens \emph{almost surely} if it has probability
$1-2^{-n}\poly(n)$.

We denote by $[m]$ the set $\{1, \ldots, m\}$.  The density of a
set~$S \subseteq \{0,1\}^n$ is $\mu(S) = \frac{|S|}{2^n}$.  We
sometimes identify a set~$S$ with its characteristic function $S:
\{0,1\}^n \to \{0,1\}$.  We often denote a tuple $(x_1, x_2, \ldots,
x_k)$ by $x^{(k)}$.

If a distribution $\mu$ over some set is given, we write~$x\leftarrow
\mu$ to denote that $x$ is chosen according to $\mu$.  We sometimes
identify sets with the uniform distribution over them.
We let $\mu_\delta$ be the Bernoulli distribution over~$\{0,1\}$ with
parameter~$\delta$, i.e., $\Pr_{x\leftarrow \mu_\delta}[x=1] =
\delta$.  Furthermore, $\mu_\delta^k$ is the distribution
over~$\{0,1\}^k$ where each bit is i.i.d.\ according to $\mu_\delta$.

When two interactive algorithms $A$ and $B$ are given, we will denote
by $\<A,B\>_A$ the output $A$ has in an interaction with $B$, and by
$\<A,B\>_B$ the output which $B$ has.  We sometimes consider
probabilities like $\Pr[\<A,B\>_A = \<A,B\>_B]$, in which case the
probability is over random coins of $A$ and $B$ (if any), but they are
chosen the same on the left and the right hand side.

\section{Efficiently Samplable Predicates}
\subsection{Single Instance}
\subsubsection{Informal Discussion}
Fix a predicate $P: \{0,1\}^n \to \{0,1\}$ and a circuit $C(x, b, r)$
which takes an arbitrary $x \in \{0,1\}^n$, a bit $b \in \{0,1\}$, and
some randomness $r$ as input.  We may think of $C$ as a circuit which
tries to distinguish the case $b = P(x)$ from the case $b = 1-P(x)$.
Our idea is to identify a set $S$ for which we can show the following:
\begin{enumerate}
\item If $x$ is picked randomly from $S$, then $\Pr[C(x, P(x), r) = 1] \approx \Pr[C(x, 1-P(x), r)=1]$.
\item $C$ can be used to predict $P(x)$ for a uniform random $x$
  correctly with probability close to $1 - \frac{1}{2} \mu(S)$
\end{enumerate}
On an informal level, one could say that $S$ explains the hardness of
computing $P$ from $C$'s point of view: for elements from $S$ the
circuit just behaves as a uniform random guess, on the others it
computes (or, more accurately, \emph{helps} to compute) $P$.  Readers
familiar with Impagliazzo's hardcore lemma will notice the similarity:
Impagliazzo finds a set which explains the computational difficulty of
a predicate for \emph{any} circuit of a certain size.  Thus, in this
sense Impagliazzo's theorem is stronger.  The advantage of ours is
that the proof is technically simpler, and that it can be used in the
interactive setting (see Section~\ref{sec:interactionPredicate}) which
seemingly comes from the fact that it helps to build non-rewinding proofs.

\subsubsection{The Theorem}

The following theorem formalizes the above discussion.  It will find
$S$ by producing a circuit which recognizes it, and also produces a
circuit $Q$ which uses $C$ in order to predict $P$.

\begin{theorem} \label{'theorem single pred'} Let $P:
  \{0,1\}^n\to\{0,1\}$ be a computable predicate.  There is an
  algorithm $\Gen$ which takes as input a randomized circuit
  $C(x,b,r)$ and a parameter $\epsilon$, and outputs two deterministic
  circuits $Q$ and $S$, both of size $\size(C)\cdot\poly(n,
  \frac{1}{\epsilon})$, as well as $\delta \in [0,1]$, such that
  almost surely the following holds:
  \begin{description}
  \item[Large Set:] $S(x,P(x))$ recognizes a set $S^* = \{x |
    S(x,P(x)) = 1\}$ of density at least $\mu(S^*) \geq \delta$.
  \item[Indistinguishability:] For the above set $S^*$ we have
    \begin{align}
      \bigl|\Pr_{x\leftarrow \{0,1\}^n, r}[C(x,P(x),r)=1] -
      \Pr_{x\leftarrow \{0,1\}^n, r}[C(x,P'(x),r)=1]\bigr| \leq \epsilon,
    \end{align}
    where $P'(x) := P(x) \oplus S(x)$, i.e., $P'$ is the predicate
    which equals~$P$ outside~$S$ and differs from~$P$ within~$S$.
  \item [Predictability:] $Q$ predicts $P$ well: $\displaystyle
    \Pr_{x\leftarrow\{0,1\}^n}[Q(x) = P(x)]
    \geq 1-\frac{\delta}{2}$.
  \end{description}
  Additionally, these algorithms have the following
  properties:
  \begin{enumerate}
    \ifshort\setlength{\itemsep}{-5pt}\fi
  \item Unless $\delta = 1$ algorithm $Q$ predicts slightly
    better:\footnote{This implies that $\delta \geq
      \frac{\epsilon}{2}$, which can always be guaranteed.}
    $\Pr[Q(x)= P(x)]\geq 1-\frac{\delta}{2} + \frac{\epsilon}{4}$.
  \item If $P$ is efficiently samplable (i.e., pairs $(x,P(x))$ can be
    generated in polynomial time), $\Gen$ runs in time $\poly(n,
    \frac{1}{\epsilon})$.
  \item $\Gen$, $S$, and $Q$ can be implemented with oracle access to
    $C$ only (i.e., they do not use the description of $C$).
  \item When thought as oracle circuits, $S$ and $Q$ use the oracle
    $C$ at most $\cO(\frac{n}{\epsilon^2})$ times.  Also, they both
    treat $x$ obliviously, and their output only depends on the number
    of $1$'s obtained from the oracle calls to $C$ and, in case of
    $S$, the input $P(x)$.
  \end{enumerate}
\end{theorem}


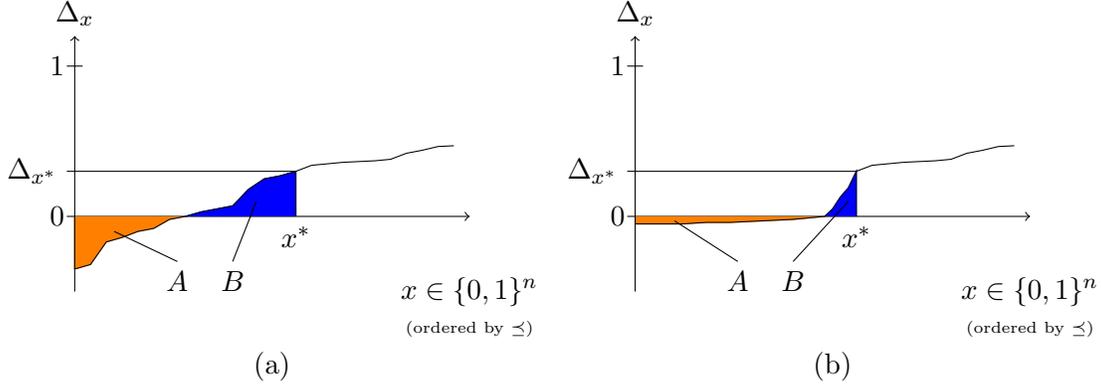
\begin{figure}
  \centering
    \begin{tikzpicture}[yscale=2,xscale=1.05]
      \draw[->]
                (0,-0.5) -- (0,0) node [left] {$0$} --
                (0,1) node [left] {$1$} -- (0,1.2) node [above] {$\Delta_x$};
      \draw[-]  (-0.1,1) -- (0.1,1);
      \draw[->] (-0.1,0) -- (5,0);
      \draw (5, -0.5) node {$x \in \{0, 1\}^n$};
      \draw (5, -0.75) node {\tiny(ordered by~$\preceq$)};

      \draw[fill, orange]
               (0,0) --
               (0.0, -0.35) -- (0.2, -0.32) -- (0.4,-.17) -- (0.6,-.14) -- (0.8,-.1) --
               (1.0, -.08 ) -- (1.2,  -.02) -- (1.4,0 ) ;
      \draw[-] (0,0) --
               (0.0, -0.35) -- (0.2, -0.32) -- (0.4,-.17) -- (0.6,-.14) -- (0.8,-.1) --
               (1.0, -.08 ) -- (1.2,  -.02) -- (1.4,0 ) ;

      \draw[fill, blue]
               (1.4,0) -- (1.6,0.03 ) -- (1.8,0.05 ) --
               (2.0,0.07) -- (2.2,0.18) --
               (2.2,0.18) -- (2.4,0.25) -- (2.6,0.27) -- (2.8,.30) -- (2.8,0) node [black,below] {$x^*$};
      \draw[-]
               (1.4,0) -- (1.6,0.03 ) -- (1.8,0.05 ) --
               (2.0,0.07) -- (2.2,0.18) --
               (2.2,0.18) -- (2.4,0.25) -- (2.6,0.27) -- (2.8,.30) -- (2.8,0);

      \draw[-] (2.8,.3) --
               (3.0,.34 ) -- (3.2,.35 ) -- (3.4,.36 ) -- (3.6,.365 ) -- (3.8,.37) --
               (4.0,.38 ) -- (4.2,.42 ) -- (4.4,.44 ) -- (4.6,.465 ) -- (4.8,.47);
      \draw[-] (2.8,.3) -- (-.1,.3) node [left] {$\Delta_{x^*}$};

      \draw[-] (0.5,-0.1) -- (1.3, -0.3) node[below] {$A$};
      \draw[-] (2.3,0.1) -- (2.0, -0.3) node[below] {$B$};

      \draw (2.5, -1) node {(a)};
    \end{tikzpicture}
    \begin{tikzpicture}[yscale=2,xscale=1.05]
      \draw[->]
                (0,-0.5) -- (0,0) node [left] {$0$} --
                (0,1) node [left] {$1$} -- (0,1.2) node [above] {$\Delta_x$};
      \draw[-]  (-0.1,1) -- (0.1,1);
      \draw[->] (-0.1,0) -- (5,0);
      \draw (5, -0.5) node {$x \in \{0, 1\}^n$};
      \draw (5, -0.75) node {\tiny(ordered by~$\preceq$)};

      \draw[fill, orange]
               (0,0) --
               (0.0, -0.05) -- (0.3, -0.05) -- (0.6,-.05) -- (0.9,-.04) -- (1.2,-.04) --
               (1.6, -.03 ) -- (2.0,-.02) -- (2.4, 0) ;
      \draw[-] (0,0) --
               (0.0, -0.05) -- (0.3, -0.05) -- (0.6,-.05) -- (0.9,-.04) -- (1.2,-.04) --
               (1.6, -.03 ) -- (2.0,-.02) -- (2.4, 0) ;

      \draw[fill, blue]
               (2.4,0) -- (2.5,0.05 ) -- (2.6,0.13) --
               (2.7,0.19) -- (2.8,.30) -- (2.8,0) node [black,below] {$x^*$};
      \draw[-]
               (2.4,0) -- (2.5,0.05 ) -- (2.6,0.13) --
               (2.7,0.19) -- (2.8,.30) -- (2.8,0);

      \draw[-] (2.8,.3) --
               (3.0,.34 ) -- (3.2,.35 ) -- (3.4,.36 ) -- (3.6,.365 ) -- (3.8,.37) --
               (4.0,.38 ) -- (4.2,.42 ) -- (4.4,.44 ) -- (4.6,.465 ) -- (4.8,.47);
      \draw[-] (2.8,.3) -- (-.1,.3) node [left] {$\Delta_{x^*}$};

      \draw[-] (0.5,-0.03) -- (1.3, -0.3) node[below] {$A$};
      \draw[-] (2.7,0.1) -- (2.0, -0.3) node[below] {$B$};
      \draw (2.5, -1) node {(b)};
    \end{tikzpicture}
    \caption{Intuition for the proof of Theorem \ref{'theorem single
        pred'}.  In both pictures, on the vertical axis, the advantage
      of the circuit in guessing right over a random guess is
      depicted.  The elements are then sorted according to this
      quantity.  The point $x^*$ is chosen such that the area of $A$
      is slightly smaller than the area of~$B$ (as in
      equation~(\ref{eq:3})).}
    \label{fig:littlepicture}
\end{figure}

\iffull
Before we give the proof, we would like to mention that
the proof uses no new techniques.  For example, it is very similar to
Lemma 2.4 in \cite{Holens05}, which in turn is implicit in
\cite{Levin87, GoNiWi95} (see also Lemma 6.6 and Claim 7 on page 121
in \cite{Holens06}).  Our main contribution here is to give the
statement and to note that it is very powerful.
\fi
\ifshort The proof uses no new techniques.  For example, it is very
similar to Lemma 2.4 in \cite{Holens05}, which in turn is implicit in
\cite{Levin87, GoNiWi95} (see also Lemma 6.6 and Claim 7 on page 121
in \cite{Holens06}).  Our main contribution here is to give the
statement and to note that it is very powerful.  The proof itself is only given in the appendix of 
the paper.  It is only remarkable for how straight-forward
it is (given the statement).  
\fi

\paragraph{Proof Overview.}
We assume that overall $C(x,P(x),r)$ is more often $1$ than
$C(x,1-P(x),r)$.  Make $S$ the largest set for which the
Indistinguishability property is satisfied as follows: order the
elements of $\{0,1\}^n$ according to $\Delta_x :=
\Pr_{r}[C(x,P(x),r)=1] - \Pr_{r}[C(x,1-P(x),r)=1]$, and insert them
into $S$ sequentially until both $\Pr_{x\leftarrow S,
  r}[C(x,P(x),r)=1] > \Pr_{x\leftarrow S, r}[C(x,1-P(x),r)=1]$ and
indistinguishability is violated.  Then, it only remains to
describe~$Q$.  For any $x\notin S$ note that $\Pr[{C(x,P(x), r)=1}] -
\Pr[{C(x,1-P(x), r)=1}] \geq \epsilon$, as otherwise $x$ could be
added to $S$.  Thus, for those elements $P(x)$ is the bit $b$ for
which $\Pr[C(x,b,r)=1]$ is bigger.  In this overview we assume that
$\Pr[{C(x, b, r)=1}]$ can be found exactly, so we let $Q(x)$ compute
the probabilities for $b = 0$ and $b = 1$, and answer accordingly; we
will call this rule the ``Majority Rule''.  Clearly, $Q(x)$ is correct
if $x \notin S$, and in order to get ``predictability'', we only need
to argue that $Q$ is not worse than a random guess on $S$.

Consider now Figure~\ref{fig:littlepicture} (a), where the elements
are ordered according to $\Delta_x$.  The areas depicted $A$ and $B$
are roughly equal, which follows by the way we chose $S$ (note
that $\Pr_{x \leftarrow S,r}[C(x,P(x),r)=1]-\Pr_{x \leftarrow
  S,r}[C(x,1-P(x),r)=1] = \E_{x \leftarrow S}[\Delta_x]$).

At this point our problem is that the majority rule will
give the incorrect answer for all elements for which $\Delta_x < 0$,
and as shown in Figure~\ref{fig:littlepicture} (b), this can be almost
all of $S$, so that in general the above $Q$ \emph{does} perform worse
than a random guess on $S$.  The solution is to note that it is sufficient to
follow the majority rule in
 case the gap is bigger than $\Delta_{x^*}$.  In the full proof we
will see that if the gap is small so that $-\Delta_{x^*} \leq \Pr[C(x,0,r)=1] - \Pr[C(x,1,r)=1] \leq \Delta_{x^*}$
then a randomized decision works: the probability
of answering $b=0$ is $1$ if the gap is $ -\Delta_{x^*}$, the probability of answering
$b=0$ is 0 if the gap is $ \Delta_{x^*}$.  When the gap is in between then the probability of answering
$b=0$ is linearly interpolated based on the value of the gap.  So for example, if the gap is 0, then $b=0$ with probability 
$\frac{1}{2}$.\footnote{It may be instructive
  to point out another rule which does not work: if one produces a
  uniform random bit in case the gap is smaller than $\Delta_{x^*}$
  then elements in the region marked $A$ with negative gap larger than
  $\Delta_{x^*}$ are problematic.}  A bit of thought reveals that this
is exactly because the areas $A$ and $B$ in
Figure~\ref{fig:littlepicture} are almost equal.

In the full proof, we also show how to sample all quantities
accurately enough (which is easy) and how to ensure that~$S$ is a set
of the right size (which seems to require a small trick because
$\Delta_x$ as defined above is not computable exactly, and so we actually use a different quantity for $\Delta_x$).  We think that
the second is not really required for the applications later, but it
simplifies the statement of the above theorem and makes it somewhat
more intuitive.

\begin{proof}
  We describe algorithm $\Gen$.  First, obtain an estimate
  \begin{align}
    \Delta :\approx \Pr_{r,x}[C(x,P(x), r)=1] -
    \Pr_{r,x}[C(x,1-P(x), r)=1]
  \end{align}
  such that almost surely $\Delta$ is within~$\epsilon/4$ of the
  actual quantity.  If~$|\Delta| < 3\epsilon/4$, we can return~$\delta
  = 1$, $S = \{0,1\}^n$, and a circuit~$Q$ which guesses a uniform
  random bit.  If~$\Delta < -3\epsilon/4$ replace $C$ with the circuit
  which outputs~$1-C$ in the following argument.  Thus, from now on
  assume~$\Delta > 3\epsilon/4$ and that the actual quantity is at
  least~$\epsilon/2$.

  Sample random strings~$r_1,\ldots,r_m$ for $C$, where~$m =
  100n/\epsilon^2$, and let~$C'(x,b,i)$ be the circuit which computes
  $C(x,b,r_i)$.  Using a Chernoff bound, we see that for all $x \in \{0,1\}^n$
  \begin{align}
    \Pr_{r}[C&(x,P(x), r)=1] -
    \Pr_r[C(x,1-P(x), r)]=1]= \nonumber\\
    & \Pr_{i \in [m]}[C'(x,P(x), i)=1] - \Pr_{i \in [m]}[C'(x,1-P(x), i)]=1]
    \pm \epsilon/4
  \end{align}
  almost surely.

  Define, for any~$x$,
  \begin{align}\label{eq:4}
    \Delta_x := \Pr_{i \in [m]}[C'(x,P(x), r_i)=1] - \Pr_{i \in [m]}[C'(x,1-P(x), r_i)=1].
  \end{align}
  Because we define~$\Delta_x$ using~$C'$ instead of $C$, we can
  compute~$\Delta_x$ exactly for a given~$x$.  Now, order the~$x$
  according to $\Delta_x$: let~$x_1 \preceq x_2$ if $\Delta_{x_1} <
  \Delta_{x_2}$, or both~$\Delta_{x_1} = \Delta_{x_2}$ and~$x_1 \leq_L
  x_2$, where~$\leq_L$ is the lexicographic ordering on bitstrings.
  We can compute~$x_1 \preceq x_2$ efficiently given $(x_1,P(x_1))$
  and $(x_2, P(x_2))$.

  We claim that we can find~$x^*$ such that almost surely (we
  assume~$\epsilon > 10\cdot 2^{-n}$, otherwise we can get the theorem
  with exhaustive search)
  \begin{align}\label{eq:3}
    \frac{\epsilon}{20} \leq \frac{1}{2^n}
    \sum_{x\preceq x^*} \Delta_x \leq \frac{\epsilon}{10}.
  \end{align}
  We pick~$50 n/\epsilon$ candidates, then almost surely one of them
  satisfies~(\ref{eq:3}) with a safety margin of~$\epsilon/50$.  For
  each of those candidates we estimate $\frac{1}{2^n} \sum_{x\preceq
    x^*} \Delta_x$ up to an error of~$\epsilon/100$, and keep one for
  which almost surely~(\ref{eq:3}) is satisfied.  We let $S(x,P(x))$
  be the circuit which recognizes the set~$S^* := \{x | x \preceq
  x^*\}$, estimate~$\delta' := |S^*|/2^n$ almost surely within an
  error of~$\epsilon/1000$, and output~$\delta := \delta' -
  \epsilon/1000$.  The situation at this moment is illustrated in
  Figure~\ref{fig:littlepicture}, and it is clear that the properties
  ``large set'' and ``indistinguishability'' are satisfied.

  We next describe~$Q$.  On input $x$, $Q$ calculates (exactly)
  \begin{align}
    \Pr_{i \in [m]}[C'(x,1, i) = 1] - \Pr_{i \in [m]}[C'(x,0, i) = 1] = (2P(x)-1)\Delta_x\;.
  \end{align}
  If $(2P(x) - 1)\Delta_x \geq \Delta_{x^*}$ (where~$\Delta_{x^*}$ is
  defined by~(\ref{eq:4}) for the element~$x^*$ which defines~$S$),
  then output~$1$, if $(2P(x) - 1)\Delta_x \leq - \Delta_{x^*}$ output
  $0$.  If neither of the previous cases apply, output~$1$ with
  probability $\frac{1}{2}(1+\frac{(2P(x) -
    1)\Delta_x}{\Delta_{x^*}})$.

  To analyze the success probability of $Q$, we distinguish two cases.
  If~$x\notin S$, we know that~$\Delta_{x} \geq \Delta_{x^*}$.
  Therefore, in this case, we get the correct answer with probability
  $1$.  If~$x \in S$, it is also easy to check that this will give the
  correct answer with probability
  $\max\{\frac{1}{2}(1+\frac{\Delta_x}{\Delta_{x^*}}), 0\}$, and thus,
  on average $\frac{1}{|S|}\sum_{x \in S}
  \max\{\frac{1}{2}(1+\frac{\Delta_x}{\Delta_{x^*}}), 0\} \geq
  \frac{1}{|S|}\sum_{x \in S}
  \frac{1}{2}(1+\frac{\Delta_x}{\Delta_{x^*}}) \geq
  \frac{\epsilon}{20}$, using~(\ref{eq:3}).  In total, we have
  probability at least $\mu(S)(\frac{1}{2}+\frac{\epsilon}{40}) +
  (1-\mu(S))$ of answering correctly.  Since $\mu(S) \geq \delta$,
  this quantity is at least $1 - \frac{\delta}{2}$, which implies
  ``predictability''.  It is possible to make $Q$ deterministic by
  trying all possible values for the randomness and estimating the
  probability of it being correct.

  In order to get the additional property 1, we first run the
  above algorithm with input $\epsilon/3$ instead of $\epsilon$.  If
  $\delta > 1 - 2\epsilon/3$, we instead output the set containing all
  elements and return $1$ in place of $\delta$.  Note that
  indistinguishability still holds because we only add a fraction of
  $2\epsilon/3$ elements to $S$.  If $\delta \leq 1 - 2\epsilon/3$, we
  enlarge $S$ by at least $\epsilon/2$ and at most $2\epsilon/3$; this
  can be done by finding a new candidate for $x^*$ as above.  We then
  output the new set and $\delta' := \delta + \frac{\epsilon}{2}$.

  The additional properties 2, 3 and 4 follow by inspection of the proof.
\end{proof}

\subsection{Multiple instances}
\subsubsection{Informal Discussion}
We explain our idea on an example: suppose we want to prove Yao's
XOR-Lemma.  Thus, we are given a predicate $P: \{0,1\}^n \rightarrow
\{0,1\}$ which is somewhat hard to compute, i.e., $\Pr[C^{(1)}(x) =
P(x)] < 1-\frac{\delta}{2}$ for any circuit $C^{(1)}$ coming from some
family of circuits (the superscript $(1)$ should indicate that this is
a circuit operating on a single instance).  We want to show that any
circuit $C^{(\oplus k)}$ from a related family predicts
$P(x_1)\oplus\dots\oplus P(x_k)$ from $(x_1,\ldots,x_k)$ correctly
with probability very close to $\frac{1}{2}$, and aiming for a
contradiction we now assume that a circuit $C^{(\oplus k)}$ exists which does
significantly better than this is given.

As a first step, we transform $C^{(\oplus k)}$ into a circuit
$C^{(k)}(x_1,b_1,x_2,b_2,\ldots,x_k,b_k)$ as follows: $C^{(k)}$
invokes $C^{(\oplus k)}(x_1,\ldots,x_k)$ and outputs $1$ if the result
equals $b_1 \oplus \dots \oplus b_k$, otherwise it outputs $0$.  We
see that we would like to show
$\Pr[C^{(k)}(x_1,P(x_1),\ldots,x_{k},P(x_k))=1] \approx \frac{1}{2}$.

Here is the key idea: we apply Theorem~\ref{'theorem single pred'}
sequentially on every position $i$ of $C^{(k)}$.  Done properly, in
each position one of the following happens: (a) we
can use $C^{(k)}$ to predict~$P(x)$ from~$x$ with probability at least
$1-\frac{\delta}{2}$, or (b) we find a large set $S^*_i$ such that if
$x_i \in S^*_i$, $C^{(k)}$ behaves roughly the same in case $b_i$
equals $P(x_i)$ and in case $b_i$ is a uniform random bit.  If (a)
happens at any point we get a contradiction and are done, so consider
the case that (b) happens $k$ times.  Recall now how $C^{(k)}$ was
built from $C^{(\oplus k)}$: it compares the output of $C^{(\oplus
  k)}$ to $b_1 \oplus \dots \oplus b_k$.  If $x_i$ lands in the large
set for any $i$ we can assume that $b_i$ is a random bit (and it is
very unlikely that this happens for no $i$).  Then, $C^{(k)}$ outputs
$1$ exactly if $C^{(\oplus k)}$ correctly predicts a uniform random
bit which is independent of the input to $C^{(\oplus k)}$.  The
probability such a prediction is correct is exactly~$\frac{1}{2}$, and
overall we get that $C^{(\oplus k)}$ is correct with probability close
to $\frac{1}{2}$.

\iffull
The theorem gives the formal statement for $C^{(k)}$, we later do the
transformation to $C^{(\oplus k)}$ as an example.
\fi
\ifshort
The theorem gives the formal statement for $C^{(k)}$, in the appendix the
transformation to $C^{(\oplus k)}$ is done as an example.
\fi

\subsubsection{The Theorem}
Fix a predicate $P: \{0,1\}^n \to \{0,1\}$ and a boolean circuit
$C^{(k)}(x_1,b_1,\ldots,x_k,b_k)$.  We are interested in the
probability that the circuit outputs $1$ in the following
Experiment~$1$:

\begin{lstlisting}
*\textbf{Experiment 1:}*
        $\forall i \in \{1,\ldots,k\}: x_i \leftarrow \{0,1\}^n$
        $\forall i\in \{1,\ldots,k\}: b_i := P(x_i)$
        $r \leftarrow \{0,1\}^*$
        output $C^{(k)}(x_1,b_1,\ldots,x_k,b_k,r)$
\end{lstlisting}

We will claim that there are large sets $S^*_1, \ldots, S^*_k$ with
the property that for any $x_i$ which falls into $S^*_i$, we can set
$b_i$ to a random bit and the probability of the experiment producing
a $1$ will not change much.  However, we will allow the sets $S^*_i$
to depend on the $x_j$ and $b_j$ for $j < i$; we therefore assume that
an algorithm $\GenS$ is given which produces such a set on input
$t_i = (x_1,b_1,\ldots,x_{i-1},b_{i-1})$.

\begin{lstlisting}
*\textbf{Experiment 2:}*
        for $i := 1$ to $k$ do
           $t_i := (x_1,b_1,\ldots,x_{i-1}, b_{i-1})$
           $S^*_i := \GenS(t_i)$
           $x_i \leftarrow \{0,1\}^n$
           if $x_i \in S^*_{i}$ then $b_i \leftarrow \{0,1\}$ else $b_i := P(x_i)$ fi
        end for
        $r \leftarrow \{0,1\}^*$
        output $C^{(k)}(x_1,b_1,\ldots,x_k,b_k,r)$
\end{lstlisting}

Theorem~\ref{'theorem many pred'} essentially states the following:
assume no small circuit can predict $P(x)$ from $x$ with probability
$1-\frac{\delta}{2}$.  For any fixed circuit~$C^{(k)}$, any
$\epsilon$, and any $k$ there is an algorithm $\GenS$ which produces
sets $S_i^*$ with $\mu(S^*_i) \geq \delta$ and such that the
probability that Experiment~1 outputs~1 differs by at most $\epsilon$
from the probability that Experiment~2 outputs~1.

\begin{theorem} \label{'theorem many pred'} Let~$P$ be a computable
  predicate, $k, \frac{1}{\epsilon} \in \poly(n)$ parameters.  There are
  two algorithms $\Gen$ and $\GenS$ as follows: $\Gen$ takes as input
  a randomized circuit $C^{(k)}$ and a parameter $\epsilon$ and
  outputs a deterministic circuit $Q$ of size
  $\size(C^{(k)})\cdot\poly(n)$ as well as
  $\delta \in [0,1]$.  $\GenS$ takes as input a circuit $C^{(k)}$, a
  tuple $t_i$, and a parameter $\epsilon$ and outputs a deterministic
  circuit $S_{t_i}(x,b)$ of
  $\size(C^{(k)})\cdot\poly(n)$.  After a run of
  $\Gen$, almost surely the following properties are satisfied:
  \begin{description}
  \item[Large Sets:] For any value of $t_i :=
    (x_1,b_1,\ldots,x_{i-1},b_{i-1})$ the circuit $S_{t_i}(x_{i},
    P(x_i))$ recognizes a set $S^*_{i} := \{x_i | S(t_i,x_i,
    P(x_i))=1\}$.  The probability that in an execution of
    Experiment~2 we have~$\mu(S^*_i) < \delta$ for any of the~$S^*_i$ which
    occur is at most~$\epsilon$.
  \item[Indistinguishability:] Using sets $S_{t_i}^*$ as above in
    Experiment 2 gives
    \begin{align}\label{eq:6}
      \bigl|\Pr[\text{Experiment 1 outputs 1}] - \Pr[\text{Experiment
        2 outputs 1}\bigr| \leq \epsilon.
    \end{align}
  \item [Predictability:] $Q$  predicts $P$ well: $
    \displaystyle\Pr_{x\leftarrow\{0,1\}^n}[Q(x) = P(x)] \geq 1-\frac{\delta}{2}$.
  \end{description}
  Additionally, these algorithms have the following
  properties:
  \begin{enumerate}
    \ifshort\setlength{\itemsep}{-5pt}\fi
  \item Unless $\delta = 1$ algorithm $Q$ predicts slightly
    better: $\Pr[Q(x) = P(x)] \geq 1-\frac{\delta}{2} +
    \frac{\epsilon}{16k}$.
  \item If $P$ is efficiently samplable (i.e., pairs $(x,P(x))$ can be
    generated in polynomial time), $\Gen$ and $\GenS$ run in time
    $\poly(n)$.
  \item $\Gen$, $\GenS$, $S_{t_i}$, and $Q$ can be implemented with
    oracle access to $C$ only (i.e., they don't use the description of
    $C$).
  \item When thought of as oracle circuits, $S_{t_i}$ and $Q$ use the
    oracle $C$ at most $\cO(\frac{k^2n}{\epsilon^2})$ times.  Also, they
    both treat $x$ obliviously and are non-rewinding.  Finally, their
    output only depends on the number of $1$'s obtained from the
    oracle calls to $C$ and, in case of $S_{t_i}$, the input $P(x)$.
  \end{enumerate}
\end{theorem}

\ifshort The proof is given in the appendix, but follows the
informal discussion above.  \fi

\iffull
\begin{proof}
  For any fixed tuple $t_i = (x_1,b_1,\ldots,x_{i-1},b_{i-1})$,
  consider the circuit~$C_{t_i}(x,b,r)$ which uses $r$ to pick
  random~$x_j$ for~$j > i$, and runs~$C^{(k)}(t_i,x,b, x_{i+1},
  P(x_{i+1}),\ldots,x_{k},P(x_k))$.\footnote{Formally, $C_{t_i}$ may
    not be a small circuit because at this point we do not assume $P$
    to be efficiently samplable, and $C_{t_i}$ seems to need to use
    $r$ to sample pairs $(x_j, P(x_j))$ for $j > i$.  However, we can
    think of $C_{t_i}$ as oracle circuit with oracle access to $P$ at
    this moment.  Inspection of the previous proof shows that later we
    can remove the calls to $P$, as the $x_{j}$ with $j > i$ can be
    fixed.}  We let $\GenS$ be the algorithm which invokes $\Gen$ with
  parameter $\frac{\epsilon}{4k}$ from Theorem~\ref{'theorem single
    pred'} on the circuit~$C_{t_i}$ and then returns the circuit
  recognizing a set from there.

  We next describe $\Gen$: For~$\ell = n k/\epsilon$ iterations, pick
  a random~$i \in \{0, \ldots, k-1\}$, use the procedure in
  Experiment~2 until loop $i$, and run algorithm $\Gen$ from
  Theorem~\ref{'theorem single pred'} with parameter
  $\frac{\epsilon}{4k}$.  This yields a parameter $\delta$ and a
  circuit $Q$.  We output the pair~$(Q,\delta)$ for the
  smallest~$\delta$ ever encountered.  Since $k$ and $\epsilon$ are
  polynomial in $n$, almost surely every time Theorem~\ref{'theorem
    single pred'} is used the almost surely part happens.  Thus, we
  get the property ``predictability'' (and in fact the stronger
  property listed under additionally) immediately.  We now argue
  ``large sets'': consider the random variable $\delta$ when we pick a
  random $i$, simulate an execution up to iteration $i$ of
  Experiment~2, then run $\Gen$ from Theorem~\ref{'theorem single
    pred'}.  Let $\delta^*$ be the $\frac{\epsilon}{k}$-quantile of
  this distribution, i.e., the smallest value such that with
  probability $\frac{\epsilon}{k}$ the value of $\delta$ is at most
  $\delta^*$.  The probability that a value not bigger than $\delta^*$
  is output by $\Gen$ is at least $1-(1-\frac{\epsilon}{k})^\ell > 1 -
  2^{-n}$, in which case ``large sets'' is satisfied.

  We show ``indistinguishability'' with a standard hybrid argument.
  Consider the Experiment~$H_j$:
\begin{lstlisting}
*\textbf{Random Experiment $H_j$:}*
        for $i := 1$ to $k$ do
           $t_i := (x_1,b_1,\ldots,t_{i-1}, b_{i-1})$
           $S_{i}^* := \GenS(t_i)$
           $x_i \leftarrow \{0,1\}^n$
           if $i \leq j$ and $x_i \in S^*_t$ then
              $b_i \leftarrow \{0,1\}$
           else
              $b_i := P(x_i)$
           end if
        end for
        $r \leftarrow \{0,1\}^*$
        output $C^{(k)}(x_1,b_1,\ldots,x_k,b_k,r)$
\end{lstlisting}

Experiment~$H_0$ is equivalent to Experiment~$1$, Experiment $H_{k}$
is the same as Experiment~$2$.  Applying Theorem~\ref{'theorem single
  pred'} we get that for every fixed~$x_1,\ldots,x_{j-1},
b_1,\ldots,b_{j-1}$, almost surely
\begin{align}
  \Bigl|
  \Pr_{x_i,\ldots,x_k}[&C^{(k)}(x_1,b_1,\ldots,x_{j-1},b_{j-1},x_j,b_j^{(j-1)},\ldots,x_k,P(x_k)) = 1] -\nonumber\\
  \Pr_{x_i,\ldots,x_k}[&C^{(k)}(x_1,b_1,\ldots,x_{j-1},b_{j-1},x_j,b_j^{(j)},\ldots,x_k,P(x_k)) = 1] \Bigr|
  \leq \epsilon/4k\;,
\end{align}
where $b_j^{(j-1)}$ is chosen as $b_j^{(j-1)} = P(x_j)$ in experiment
$H_{j-1}$, and $b_{j}^{(j)}$ is chosen the same way as $b_j$ is chosen
in experiment $H_{j}$ (in Theorem~\ref{'theorem single pred'} the bit
is flipped, but when using a uniform bit instead of flipping it the
distinguishing probability only gets smaller).  Applying the triangle
inequality $k-1$ times we get that almost surely the difference of the
probabilities in Experiment~1 and Experiment~2 is at most
$\frac{\epsilon}{2}$.  Since ``almost surely'' means with
probabilities $1-2^{-n}\poly(n) > 1-\frac{\epsilon}{2}$, we get
``indistinguishability''.

We already showed the additional Property 1.  Properties 2,3, and 4
follow by inspection.
\end{proof}
\fi

\iffull
\subsection{Example: Yao's XOR-Lemma}

As a first example, we prove Yao's XOR-Lemma from
Theorem~\ref{'theorem many pred'}.  We will give the proof for the
non-uniform model, but in fact it would also work in the uniform model
of computation.
\begin{theorem}[Yao's XOR-Lemma]
  Let $P: \{0,1\}^n \to \{0,1\}$ be a predicate, such that for all
  circuits~$Q$ of size at most~$s$:
  \begin{align}\label{eq:14}
    \Pr[Q(x) = P(x)] < 1-\frac{\delta'}{2}.
  \end{align}
  Then, for all circuits of size $s/\poly(n, k, \frac{1}{\epsilon'})$:
  \begin{align}\label{eq:11}
    \Pr[C^{(\oplus k)}(x_1,\ldots,x_k) = P(x_1)\oplus \dots \oplus P(x_k)] \leq
    \frac{1}{2} + (1-\delta')^k + \epsilon'.
  \end{align}
\end{theorem}
\begin{proof}
  Assume a circuit $C^{(\oplus k)}$ which contradicts (\ref{eq:11}) is
  given, we will obtain a circuit $Q$ which contradicts (\ref{eq:14}).
  For this, let $C(x_1,b_1,\ldots,x_k, b_k)$ be the circuit which runs
  $C^{(\oplus k)}(x_1,\ldots,x_k)$ and outputs $1$ if the result is
  the same as $b_1\oplus\dots\oplus b_k$.  We apply
  Theorem~\ref{'theorem many pred'} setting the parameter $\epsilon$
  to $\epsilon'/2$, which produces (among other things) a parameter
  $\delta$.  We assume that the 3 properties which almost surely hold
  do hold (otherwise run $\Gen$ again).  In case
  $\delta < \delta'$, we use $Q$ to get
  a contradiction.  Otherwise, we get
  \begin{align}
    \Pr[C(&x_1,\ldots,x_k) = P(x_1)\oplus \dots \oplus P(x_k)]
    = \Pr[C(x_1,P(x_1),\ldots,x_k,P(x_k)) = 1]\\
    &\leq
    \Pr[\text{$C$ outputs $1$ in Experiment 2}] + \frac{\epsilon'}{2} \\
    &\leq
    \Pr[\text{$C$ outputs $1$ in Experiment 2 and all sets $S_i^*$
      were of density at least $\delta$}] + \epsilon'\\
    & \leq
    \frac{1}{2} + (1-\delta)^k + \epsilon'\;.
  \end{align}
\end{proof}

\subsection{Example: Extraction Lemma (\cite{HILL99})}

Roughly speaking, the construction of a pseudorandom generator from an
arbitrary one-way function proceeds in two steps (see \cite{Holens06b}
for a more detailed description of this view).  First, using the
Goldreich-Levin theorem \cite{GolLev89}, one constructs a
\emph{pseudo-entropy pair}\footnote{While \cite{HILL99} constructs a
  PEP implicitly, the definition and name was introduced in
  \cite{HaHaRe06b}.}  $(f,P)$, which is a pair of functions $f: \{0,
1\}^n \rightarrow \{0, 1\}^{\poly(n)}$, $P: \{0, 1\}^n \rightarrow
\{0, 1\}$ such that for all efficiently computable~$A$,
\begin{align}\label{eq:5}
  \Pr[A(f(x)) = P(x)] \leq 1-\frac{\delta'}{2},
\end{align}
for some non-negligible $\delta'$, and which satisfies some additional
information theoretic property (the information theoretic property
ensures that predicting $P(x)$ from~$f(x)$ is a computational problem,
and (\ref{eq:5}) does not already hold because $f$ is, say, a constant
function).

Second, given independently sampled instances $x_1,\ldots,x_k$, the
extraction lemma then says that extracting $(\delta'-\frac{1}{n})k$
bits from the concatenation $P(x_1)\ldots P(x_k)$ will give a string
which is computationally indistinguishable from a random string.  Due
to the information theoretic property above, once one has the
extraction lemma, it is relatively easy to get a pseudo-random
generator.  In the following we will prove this extraction lemma.

A technicality: the predicate which is hard to predict in this case is
supposed to have input $f(x)$ and output $P(x)$.  However, in reality
this does not have to be a predicate: $f$ is not always injective (in
fact, for $f$ obtained as above it will not be).  Most works avoid
this problem by now stating that previous theorems also hold for
randomized predicates.  This is often true, but some of the statements
get very subtle if one does it this way, and statements which involve
sets of ``hard'' inputs very much so.  We therefore choose to solve
the problem in a different way.  We consider circuits which try to
predict $P(x)$ from $x$, but are limited in that they first are
required to apply $f$ on $x$ and not use $x$ anywhere else.  Now, we
have a predicate again, but it can only be difficult for this
restricted class.  However, since the oracle circuit $Q$ in
Theorem~\ref{'theorem many pred'} treats $x$ obliviously we stay
within this class.

\begin{lemma}[Extraction Lemma, implicit in \cite{HILL99}]
  Let $(f,P)$ be a pair of functions satisfying~(\ref{eq:5}) for any
  polynomial time machine $A$ and set $k = 1/n^3$.  Let~$\Ext(m,s)$ be
  a strong extractor which extracts $m = (\delta'-\frac{1}{n})k$ bits
  from any $k$-bit source with min-entropy~$(\delta' - \frac{1}{2n})
  k$ such that the resulting bits have statistical distance at
  most~$2^{-n}$ from uniform.  Then, for any polynomial time~$A$
  \begin{align}
    \Pr[A(f(x_1),\ldots,f(x_k),s,\Ext(P(x_1)\cdots P(x_k), s)) = 1] -
    \Pr[A(f(x_1),\ldots,f(x_k),s,U_m) = 1]\label{eq:15}
  \end{align}
  is negligible.
\end{lemma}
\begin{proof}
  Assume otherwise, and let $\epsilon(n)$ be inverse polynomial and
  infinitely often smaller than the distinguishing advantage of $A$. We
  consider the circuit $C^{(k)}(x_1,b_1,\ldots,x_k,b_k)$ which first
  applies~$f$ on every~$x_i$, then pick $s$ at random, computes $z :=
  \Ext(b_1,\ldots,b_k,s)$, and executes $A(f(x_1),\ldots,f(x_k),s,z)$.
  We apply Theorem~\ref{'theorem many pred'} on $C$ using
  parameter~$\frac{\epsilon}{2}$, which produces, among other things,
  a parameter~$\delta$.  Consider first the case $\delta < \delta' -
  \frac{1}{4n}$.  Then, there is a circuit~$Q$ which predicts $P(x)$
  from~$x$ and uses~$x$ obliviously in~$C$.  This implies that the
  resulting circuit evaluates $f(x)$ for any input~$x$ and ignores the
  input otherwise; we can therefore strip off this evaluation, and get
  a circuit which contradicts~(\ref{eq:5}).  In case $\delta \geq
  \delta' - \frac{1}{4n}$, we run Experiment 2.  If all sets which
  occur in the experiment are of size at least $\delta$ (and this
  happens with probability at least $1-\epsilon/2$), then we can use a
  Chernoff-Bound to see that with probability $1 - 2^{-\Omega(n)}$, at
  least $(\delta - \frac{1}{4n})k \geq (\delta' - \frac{1}{2n})k$ of
  the $x_i$ land in their respective set $S_t^*$.  Thus, in this case
  the extractor will produces a~$z$ which is $2^{-\Omega(n)}$-close to
  uniform and the indistinguishability property of
  Theorem~\ref{'theorem many pred'} implies that~(\ref{eq:15}) is
  negligible.
\end{proof}

\fi
\subsection{Cryptographic Protocols which output single bits}
\label{sec:interactionPredicate}
Again we start with an example: consider a slightly weak bit
commitment protocol, where the receiver can guess the bit the
sender committed to with probability $1-\frac{\delta}{2}$.  In such a
case, we might want to strengthen the scheme.  For example, in order
to commit to a single bit~$b$, we could ask the sender to first commit
to two random bits $r_1$ and $r_2$, and then send $b \oplus r_1 \oplus
r_2$ to the receiver.  The hope is that the receiver has to guess both
$r_1$ and $r_2$ correctly in order to find $b$, and so the protocol
should be more secure.

In the case where the protocol has some defect that sometimes allows a
sender to cheat, we might also want to consider the protocol where
the sender commits twice to $b$, or, alternatively, that he commits to
$r_1$, then to $r_2$, and sends both $b \oplus r_1$ and $b \oplus r_2$
to the receiver.  In this case, one can  hope that a cheating
receiver still needs to break the protocol at least once, and that the
security should not degrade too much.

Just how will the security change?  We want to consider a scenario in
which the security is information theoretic.  We can do this by
assuming that instead of the weak protocol, a trusted party
distributes a bit~$X$ to the sender and some side information~$Z$ to
the receiver.  The guarantee is that for any $f$, $\Pr[f(Z) = X] \leq 1 -
\frac{\delta}{2}$.  In such a case, one can easily obtain bounds on the
security of the above protocols, and the hope is that the same bounds
hold in the computational case.  The theorem below states that this is
indeed true (for protocols where the security consists of hiding
single bits).

We remark that while the two aforementioned examples of protocol composition are already handled
in \cite{HalRab08} (their result applies to any direct product and any XOR as
above), Theorem~\ref{'theorem bit protocols'} handles any
information theoretic amplification protocol as long as it can be
implemented efficiently.

\begin{definition}
  A pair $(X,Z)$ of random variables over $\{0,1\} \times \cZ$, where
  $\cZ$ is any finite set, is $\delta$-hiding if
  \begin{align}
    \max_{f: \cZ \to \{0,1\}}\Pr[f(Z) = X] \leq
    1-\frac{\delta}{2}.
  \end{align}
\end{definition}

\begin{theorem}\label{'theorem bit protocols'}
  Let a cryptographic protocol (which we think of as ``weak'') $W =
  (A_W, B_W)$ be given in which $A_W$ has as input a single bit $c$.
  Assume that there is a function $\delta$ such that for any
  polynomial time adversary $B_W^*$ there is a negligible function
  $\nu$ such that
  \begin{align}\label{eq:16}
    \Pr_{x \leftarrow \{0,1\}}[\<A_W(x),B_W^*\>_{B} = x] \leq 1 -
    \frac{\delta}{2} + \nu(n),
  \end{align}
  where the probability is also over the coins of $A_W$ and $B_W^*$
  (if any).

  Let further an information theoretic protocol $I = (A_I, B_I)$ be
  given.  In $I$, $A_I$ takes $k$ input bits $(X_1,\ldots,X_k)$ and
  has a single output bit.  Furthermore, assume that $I$ is
  hiding in the sense that for $k$ independent $\delta$-hiding random
  variables $(X_i, Z_i)$, any (information theoretic) adversary
  $B_I^*$, and for some function~$\eta(k)$:
  \begin{align}\label{eq:17}
    \Pr[\<A_I(X_1,\ldots,X_k), B_I^*(Z_1,\ldots,Z_k)\>_A =
    \<A_I(X_1,\ldots,X_k), B_I^*(Z_1,\ldots,Z_k)\>_{B}]< \frac{1}{2} + \eta(k).
  \end{align}

  Let $S = (A_S,B_S)$ be the protocol where $A$ and $B$ first execute
  $k(n)$ copies of $W$ sequentially, where $A$ uses uniform random
  bits as input.  Then, they run a single execution of protocol $I$.
  In the execution to $I$, $A$ uses his $k$ input bits to the weak
  protocols as input.  The output of $A$ in $S$ is the output of $A$
  in the execution of $I$.  We also need that $(A_I,B_I)$ and $k(n)$
  are such that $I$ can be run in time $\poly(n)$ for $k = k(n)$.

  Then, for any polynomial time $B_S^*$ there is a negligible function
  $\nu'$ such that
  \begin{align}\label{eq:7}
    \Pr[\<A_S, B_S^*\>_{A} = \<A_S, B_S^*\>_{B}] \leq \frac{1}{2} +
    \eta(k) + \nu'(n)\;.
  \end{align}
\end{theorem}
\begin{proof}
  Let $x \in \{0,1\}^n$ be the concatenation of the randomness which
  $A$ uses in an execution of the protocol $W$ and his input bit $c$.
  We let $P: \{0,1\}^n \to \{0,1\}$ be the predicate which outputs $c
  = P(x)$.

  In order to obtain a contradiction, we fix an adversary $B_S^*$ for
  the protocol~$S$ which violates~(\ref{eq:7}).  We would like to
  apply Theorem~\ref{'theorem many pred'}.  For this, we define
  $C^{(k)}(x_1,b_1,\ldots,x_k,b_k)$ as follows: $C^{(k)}$ first
  simulates an interaction of $B_S^*$ with $A_S$, where $A_S$ uses
  randomness $x_i$ in the $i$th invocation of the weak protocol~$W$.
  After this, $B_S^*$ is in some state in which it expects an
  invocation of the information theoretic protocol.  $C^{(k)}$
  simulates this information theoretic protocol, but it runs $A_I$
  with inputs $b_1,\ldots,b_k$ instead of the actual inputs to the
  weak protocols.  In the end, $B_S^*$ produces a guess for the output
  bit of $A_S$, and $C^{(k)}$ outputs $1$ if this guess equals the
  output of $A_I(b_1,\ldots,b_k)$ in the simulation.

  In Experiment~1 of Theorem~\ref{'theorem many pred'}, $b_i = P(x_i)$
  is used, and so $C^{(k)}$ exactly simulates an execution of the
  protocol~$S$.  Since we assume that $B_S^*$ contradicts
  (\ref{eq:7}), we see that the probability that $C^{(k)}$
  outputs~$1$ in Experiment 1 is, for infinitely many~$n$ and some
  constant~$c$ at least $\frac{1}{2} + \eta(k) + n^{-c}$.

  We now apply Theorem~\ref{'theorem many pred'} on the circuit
  $C^{(k)}$ with parameter $n^{-c}/3$.  This yields a parameter
  $\delta_{T\ref{'theorem many pred'}}$ (the subscript indicates that
  it is from Theorem~\ref{'theorem many pred'}).  We claim
  that
  \begin{align}\label{eq:18}
    \delta_{T\ref{'theorem many pred'}} \leq \delta &&\text{almost surely.}
  \end{align}

  To see this, we assume otherwise and obtain a contradiction.  In
  Experiment 2, Let $\Gamma_i$ be the communication produced by the
  weak protocol $W$ in round $i$.  Assuming all sets $S_i^*$ in the
  execution are of size at least $\delta$ (this happens with
  probability at least $1-n^{-c}/3$), the tuples $(b_i,\Gamma_i)$
  are $\delta$-hiding random variables.  Consequently, when the
  circuit $C^{(k)}$ simulates the information theoretic protocol $I$
  using bits $b_i$, it actually simulates it in an instance in which
  it was designed to be used.  Since (\ref{eq:17}) holds for an
  arbitrary adversary in this case we get that
  \begin{align}
    \Pr[\text{$C^{(k)}$ outputs $1$ in Experiment~2} | \text{No set
      $S_i^*$ was of measure less than $\delta$}] \leq \frac{1}{2} +
    \eta(k).
  \end{align}
  Therefore, the probability that $C^{(k)}$ outputs $1$ in Experiment
  2 is at most $\frac{1}{2} + \eta(k) + \frac{n^{-c}}{3}$, and using
  ``indistinguishability'' the probability that $C^{(k)}$ outputs~$1$
  in Experiment 1 is at most $\frac{1}{2} + \eta(k) +
  \frac{2n^{-c}}{3}$.  However, our assumption was that the
  probability that $C^{(k)}$ outputs~$1$ is at least $\frac{1}{2} +
  \eta(k) + n^{-c}$, and so almost surely $\Gen$ does not
  output such a big $\delta_{T\ref{'theorem many pred'}}$, establishing
  (\ref{eq:18}).

  Theorem~\ref{'theorem many pred'} also provides us with a
  non-rewinding circuit~$Q$ which treats $x$ obliviously and which
  satisfies ``predictability''.  We explain how to use $Q$ to
  break~(\ref{eq:16}), the security property of the weak protocol $W$.

  Since $Q(x)$ is non-rewinding, it uses the input $x$ exclusively in
  a fixed position~$i$, together with a fixed prefix
  $(x_1,\ldots,x_{i-1})$, in all calls to $C^{(k)}$.  We first extract
  $i$ and the prefix.

  We now explain a crucial point: how to interact with $A_W$ in order
  to cheat.  We simulate the $i-1$ interactions of $A_W$ with $B_S^*$
  up to and including round $i-1$ using $(x_1, \ldots, x_{i-1})$ as
  the input bit and randomness of $A$.  In round~$i$, we continue with
  the \emph{actual} interaction with $A_W$.  Here, $A_W$ uses
  randomness $x$ (on which we, however, do not have access).

  After this interaction, we need to be able to extract the bit $c$ of
  $A_W$.  For this, we evaluate $Q(x)$, which we claim is possible.
  Since $Q$ is oblivious and deterministic, the only difficulty is in
  evaluating the calls to $C^{(k)}(x_1,b_1,\ldots,x_{k},b_k,r)$.  All
  calls use the same values for $x_1,\ldots,x_i$.  Recalling how
  $C^{(k)}$ is defined, we see that we can continue from the state we
  had after the interaction with $A_W$ in order to evaluate $C^{(k)}$
  completely (note that all the $b_i$ are given, so the we can also
  evaluate the information theoretic protocol $I$).

  We get from Theorem~\ref{'theorem many pred'} that $Q$ satisfies,
  almost surely, infinitely often, using (\ref{eq:18})
  \begin{align}
    \Pr_{x\leftarrow\{0,1\}^n}[Q(x) = P(x)] \geq 1 -
    \frac{\delta}{2} + \frac{1}{48 kn^c}\;.
  \end{align}
  This therefore gives a contradiction to (\ref{eq:16}): in order to
  get rid of the ``almost surely'', we just consider the algorithm
  which first runs $\Gen$ and then applies the above protocol -- this
  only loses a negligible additive term in the probability.
\end{proof}

\section{Weakly Verifiable Puzzles}

\subsection{Interactive Weakly Verifiable Puzzles}
Consider a bit commitment protocol, in which a sender commits to a
single bit $b$.  In a first phase the sender and the receiver enact in an
interactive protocol, after which the sender holds some opening
information $y$, and the receiver has some way of checking whether
$(y,b)$ is a valid decommitment.  If the protocol is secure, then it
is a computationally hard problem for the sender to come up with two
strings $y_0$ and $y_1$ such that both $(y_0,0)$ and $(y_1,1)$ are
valid decommitments, in addition, he may not even know the function the receiver will use
to validate a decommitment pair,\footnote{One might want to generalize this by
  saying that in order to open the commitment, sender and receiver
  enter yet another interactive protocol.  However, our presentation
  is without loss of generality: the sender can send the randomness he
  used in the first protocol instead.  The receiver then checks, if
  this randomness together with $b$ indeed produces the communication
  in the first round, and whether in a simulation of the second
  protocol he accepts.}  and thus in general there is no
way for the sender to recognize a valid pair $(y_0,y_1)$.  We abstract
this situation in the following definition; in it we can say that the
solver produces no output because in the security property all efficient
algorithms are considered anyhow.
\begin{definition}
  An \emph{interactive weakly verifiable puzzle} consists of a
  protocol $(P,S)$ and is given by two interactive algorithms $P$ and
  $S$, in which $P$ (the problem poser) produces as output a
  circuit~$\Gamma$, and $S$ (the solver) produces no output.

  The \emph{success probability} of an interactive algorithm~$S^*$ in solving a
  weakly verifiable puzzle $(P,S)$ is:
  \begin{align}
    \Pr[y = \<P, S^*\>_{S^*}; \Gamma(y) = 1]
  \end{align}
  The puzzle is \emph{non-interactive} if the protocol consists of $P$
  sending a single message to $S$.
\end{definition}
\medskip

Our definition of a non-interactive weakly verifiable
puzzle coincides with the usual one \cite{CaHaSt05}.  The security
property of an interactive weakly verifiable puzzle is that for any
algorithm (or circuit) $S^*$ of a restricted class, the success
probability of $S^*$ is bounded.

An important property is that~$S^*$ does not get access to~$\Gamma$.
Besides bit commitment above, an example of such a puzzle is a
CAPTCHA.  In both cases it is not obvious whether a given solution is
actually a correct solution.

\subsection{Strengthening interactive weakly verifiable puzzles}
Suppose that $g$ is a monotone boolean function with~$k$ bits of
input, and $(P^{(1)},S^{(1)})$ is a puzzle.  We can consider the
following new puzzle $(P^{(g)}, S^{(g)})$: the sender and the receiver
sequentially create~$k$ instances of $(P^{(1)}, S^{(1)})$, which
yields circuits $\Gamma^{(1)}, \ldots, \Gamma^{(k)}$ for $P$.  Then
$P^{(g)}$ outputs the circuit $\Gamma^{(g)}$ which computes
$\Gamma^{(g)}(y_1,\ldots,y_k) =
g(\Gamma^{(1)}(y_1),\ldots,\Gamma^{(k)}(y_k))$.

Intuitively, if no algorithm solves a single puzzle~$(P^{(1)},
S^{(1)})$ with higher probability than~$\delta$, the probability that
an algorithm solves~$(P^{(g)}, S^{(g)})$ should not be more than
approximately~$\Pr_{u\leftarrow \mu_\delta^k}[{g(u)=1}]$. (Recall that
$\mu_\delta^k$ is the distribution on $k$-bits, where each bit is
independent and 1 with probability $\delta$.)  The following theorem
states exactly this.

\begin{theorem}\label{'theorem wvp'}
  There exists an algorithm $\Gen(C,g,\epsilon, \delta,n)$ which takes
  as input a circuit $C$, a monotone function $g$, and parameters $\epsilon,
  \delta, n$, and produces a circuit~$D$ such that the following
  holds.  If~$C$ is such that
  \begin{align}\label{'eq-puzzle-theorem-A'}
    \Pr[\Gamma^{(g)}(\<P^{(g)}, C\>_C) = 1] \geq \Pr_{u \leftarrow
      \mu_{\delta}^{k}}[g(u) = 1] + \epsilon,
  \end{align}
  then, $D$ satisfies almost surely,
  \begin{align}\Pr[\Gamma^{(1)}(\<P^{(1)}, D\>_D) = 1] \geq \delta +
    \frac{\epsilon}{6k}.
    \label{'eq-puzzle-theorem-B'}
  \end{align}
  Additionally, $\Gen$ and~$D$ only require oracle access to both~$g$
  and~$C$, and $D$ is non-rewinding.

  Furthermore, $\size(D) \leq \size(C) \cdot
  \frac{6k}{\epsilon}\log(\frac{6k}{\epsilon})$ and $\Gen$ runs in
  time $\poly(k,\frac{1}{\epsilon},n)$ with oracle calls to $C$.
\end{theorem}
The monotone restriction on~$g$ in the previous theorem is necessary.
For example, consider $g(b) = 1-b$.  It is possible to satisfy $g$
with probability 1 by producing an incorrect answer, but $\Pr_{u
  \leftarrow \mu_\delta}[g(u) = 1] = 1 - \delta$.

\subsection{Proof of Theorem~\ref{'theorem wvp'}}
\paragraph{Algorithm Description}
If $k=1$, $\Gen$ creates the circuit $D$ which runs $C$ and outputs its
answer.  Then either $g$ is the identity or a constant function. If
$g$ is the identity, the statement is trivial.  If $g$ is a constant
function, the statement is vacuously true. $D$ is non-rewinding.

In the general case, we need some notation.  For $b \in \{0, 1\}$, let
$\cG_b$ denote the set of inputs~$\cG_b := \{b_1,\ldots,b_k |
g(b,b_2,\ldots,b_k)=1\}$ (i.e., the first input bit is disregarded and
replaced by $b$).  We remark that $\cG_0 \subseteq \cG_1$ due to monotonicity of $g$.  We will
commonly denote by $u = u_1 u_2 \cdots u_k \in \{0, 1\}^k$ an element
drawn from~$\mu_\delta^k$.  After a given interaction of $C$ with $P^{(g)}$,
let $c = c_1 c_2 \cdots c_k \in \{0, 1\}^k$ denote the string where
$c_i$ is the output of $\Gamma^{(i)}$ on input $y_i$, which is the
$i$th output of $C$.  We denote the randomness used by~$P^{(g)}$ in
execution $i$ by $\pi_i$.

For $\pi^*, b \in \{0, 1\}^n \times \{0, 1\}$ we now define the
surplus $S_{\pi^*,b}$.  It denotes how much better $C$ performs than
``it should'', in the case where the randomness of~$P^{(g)}$ in the first
instance is fixed to $\pi^*$, and the output of~$\Gamma^{(1)}(y_1)$ is
ignored (i.e., we don't care whether $C$ solves the first puzzle
right), and~$b$ is used instead:
\begin{align}\label{eq:13}
  S_{\pi^*, b} := \Pr_{\pi^{(k)}}[c \in \cG_b| \pi_1 = \pi^*] -
  \Pr_{u\leftarrow\mu_{\delta}^k}[u \in \cG_b],
\end{align}
where the first probability is also over the interaction between
$P^{(g)}$ and $C$ as well as randomness $C$ uses (if any).

The algorithm then works as follows: first
pick~$\frac{6k}{\epsilon}\log(n)$ candidates $\pi^*$ for the
randomness of $P^{(g)}$ in the first position.  For each of those,
simulate the interaction $(P^{(g)}, C)$ and then get
estimates~$\Stilde_{\pi^*,0}$ and $\Stilde_{\pi^*,1}$ of~$S_{\pi^*,0}$
and $S_{\pi^*,1}$ such that~$|\Stilde_{\pi^*,b}-S_{\pi^*,b}|\leq
\frac{\epsilon}{4k}$ almost surely.

We consider two cases:
\begin{itemize}
\item One of the estimates satisfies $\Stilde_{\pi^*,b} \geq
  (1-\frac{3}{4k})\epsilon$.

  In this case, we fix~$\pi_1 := \pi^*$ and~$c_1 := b$, and
  invoke~$\Gen(C',g',(1-\frac{1}{k})\epsilon,\delta,n)$, using the
  function $g'(b_2,\dots,b_k) = g(c_1,b_2,\dots,b_k)$ and circuit $C'$
  which is defined as follows: $C'$ first (internally) simulates an
  interaction of $P^{(1)}$ with $C$, then follows up with an
  interaction with $P^{(g')}$.
\item For all estimates $\Stilde_{x^*,b} < (1-\frac{3}{4k})\epsilon$.

  In this case, we output the following circuit~$D^C$: in a first
  phase, use $C$ to interact with $P^{(1)}$.  In the second phase,
  simulate $k-1$ interactions with $P^{(1)}$ and obtain
  $(y_1,\ldots,y_k) = C(x,x_2,\ldots,x_k)$.  For $i = 2,\ldots,k$ set
  $c_i = \Gamma_i(y_i)$. If~$c = (0,c_2,\ldots,c_k) \in \cG_1
  \setminus \cG_0$, return~$y_1$, otherwise repeat the second phase
  $\frac{6k}{\epsilon}\log(\frac{6k}{\epsilon})$ times.  If all
  attempts fail, return the special value~$\bot$ (or an arbitrary
  answer).
\end{itemize}

\ifshort Due to space constraints, the proof of correctness of the above algorithm is omitted, but can be found in the appendix.
\fi

\iffull

\paragraph{Overview of Correctness}
The interesting case is when $\Gen$ does not recurse.  In this case we
know that $C$ has higher success probability than $\Pr_{u\leftarrow
  \mu_\delta^k}[g(u)=1]$, but for most $\pi^*$, the surpluses
$S_{\pi^*,0}$ and $S_{\pi^*,1}$ are less than
$(1-\frac{1}{k})\epsilon$.  Intuitively, then $C$ is correct on the
first coordinate unusually often when $c \in \cG_1 - \cG_0$ (as this
is the only time that being correct on the first coordinate helps).
If we could assume that 1) that the algorithm \emph{always} outputs an
answer, and 2) for \emph{every} $\pi^*$, the surpluses, $S_{\pi^*,0}$
and $S_{\pi^*,1}$ are less than $(1-\frac{1}{k})\epsilon$, then the
theorem would follow by straight-forward manipulations of probability.

Unfortunately these assumptions are not true, but the proof below
shows that because these assumptions only fail slightly, not much is
lost.  Informally, Equations~\ref{eq:22}-\ref{eq:27} show that if the
algorithm fails to output an answer it is either because
$\Pr_{\pi^{(k)}}[c \in \cG_1 - \cG_0|\pi_1 = \pi^*]$ is very small (in
which case this $\pi^*$ will not contribute much anyhow), or because
we are unlucky (which happens with very small probability).
Additionally, Equations~\ref{eq:29}-\ref{eq:33} show that because we
did not find a $\pi^*$ with large surplus, we can assume that (unless
we were very unlucky) there are few $\pi^*$ with large surpluses,
which cannot have undue influence.

\paragraph{Analysis of Correctness}

Consider first the case that we find $(x^*, b)$ for which
$\Stilde_{x^*,b} \geq (1-\frac{3}{4k})\epsilon$.  We can assume that
$S_{x^*,b} \geq (1-\frac{1}{k})\epsilon$, since the error is at
most~$\epsilon/(4k)$ almost surely.  Thus, we satisfy all the
requirements to use $\Gen$ with~$k-1$ (using $x^*$ as the first input
and $g(b, \cdot)$ as the monotone function with $k$ inputs), which
will return a non-rewinding circuit for which
$\Pr_{(x,\Gamma),r}[\Gamma(D(x,r)) = 1] \geq \delta +
(1-\frac{1}{k})\epsilon/6(k-1) = \delta + \epsilon/6k$.  The remaining
properties are easily verified.

The more interesting case is if $\Gen$ does not recurse.  First, we get,
for any puzzle~$\pi^* = (x^*,\Gamma^*)$ (simply using~(\ref{eq:13})
and~$\cG_0 \subseteq \cG_1$):
\begin{align}
  \Pr_{u \leftarrow \mu_{\delta}^k}[u\in\cG_1 - \cG_0] = \Pr_{\pi^{(k)}}[c \in \cG_1 - \cG_0|\pi_1 = \pi^*] -
  (S_{\pi^*,1} - S_{\pi^*,0})
\end{align}
and thus, still fixing~$\pi^*$ and multiplying by  $\Pr_r[\Gamma^*(D(x^*,r))=1] / \Pr_{u \leftarrow \mu_{\delta}^k}[u\in\cG_1 - \cG_0]$:
\begin{align}
  \Pr_r[\Gamma^*(D(x^*,r))=1] &=
  \frac{\Pr_r[\Gamma^*(D(x^*,r))=1] \Pr_{\pi^{(k)}}[c \in \cG_1 - \cG_0|\pi_1 = \pi^*]}
  {\Pr_{u \leftarrow \mu_{\delta}^k}[u\in\cG_1-\cG_0]}  \nonumber\\
  &\qquad\qquad-
  \frac{\Pr_r[\Gamma^*(D(x^*,r))=1] (S_{\pi^*,1}-S_{\pi^*,0})}
  {\Pr_{u \leftarrow \mu_{\delta}^k}[u\in\cG_1-\cG_0]}.\label{eq:10}
\end{align}
We bound the first summand in (\ref{eq:10}):
\begin{align}  \label{eq:22}
  \Pr_r[\Gamma^*(&D(x^*,r))=1] \Pr_{\pi^{(k)}}[c \in \cG_1 - \cG_0|\pi_1 = \pi^*]\nonumber\\
  &=
  \Pr_r[\Gamma^*(D(x^*,r))\neq \bot]
  \Pr_{\pi^{(k)}}[c_1 = 1 | c \in \cG_1 - \cG_0, \pi_1 = \pi^*]
  \Pr_{\pi^{(k)}}[c \in \cG_1 - \cG_0 | \pi_1 = \pi^*].
\end{align}
If $\Pr[c\in\cG_1 - \cG_0|\pi_1 = \pi^*] \leq \frac{\epsilon}{6k}$,
then~$0 \geq \Pr[c_1 = 1 | c \in \cG_1 - \cG_0, x_1 = \pi^*] \Pr[c \in
\cG_1 - \cG_0 | \pi_1 = \pi^*] -\frac{\epsilon}{6k}$.  If
$\Pr[c\in\cG_1 - \cG_0| \pi_1 = \pi^*] > \frac{\epsilon}{6k}$ then
$\Pr[\Gamma^*(C(x^*))\neq\bot] \geq 1-\frac{\epsilon}{6k}$ since $D$
only outputs~$\bot$ if after $\frac{6k}{\epsilon}\log(6k/\epsilon)$
none of the elements~$c$ was in~$\cG_1-\cG_0$.  In both cases:
\begin{align}
  \Pr_r[\Gamma^*(&D(x^*,r))\neq \bot]
  \Pr_{\pi^{(k)}}[c_1 = 1 | c \in \cG_1 - \cG_0, \pi_1 = \pi^*]
  \Pr_{\pi^{(k)}}[c \in \cG_1 - \cG_0 | \pi_1 = \pi^*]\\
  &\geq
  \Pr_{\pi^{(k)}}[c_1 = 1 | c \in \cG_1 - \cG_0 , \pi_1 = \pi^*]
  \Pr_{\pi^{(k)}}[c \in \cG_1 - \cG_0 | \pi_1 = \pi^*] - \frac{\epsilon}{6 k}\\
  &=
  \Pr_{\pi^{(k)}}[c_1 = 1 \land c \in \cG_1 - \cG_0| \pi_1 = \pi^*] - \frac{\epsilon}{6 k}\\
  &=
  \Pr_{\pi^{(k)}}[g(c) = 1 | \pi_1 = \pi^*] - \Pr_{\pi^{(k)}}[c \in \cG_0 | \pi_1 = \pi^*] -
  \frac{\epsilon}{6 k}\\
  &=\Pr_{\pi^{(k)}}[g(c) = 1 | \pi_1 = \pi^*] - \Pr_{u \leftarrow \mu_{\delta}^k}[u \in \cG_0] - S_{\pi^*,0}
  - \frac{\epsilon}{6 k}  \label{eq:27}
\end{align}
Inserting into (\ref{eq:10}) gives
\begin{align}
  \E_{\pi^*}\bigl[ \Pr_{r}[\Gamma^*(D(x^*,r))= 1]\bigr]
  &\geq
  \E_{\pi^*}\Bigl[
  \frac{\Pr_{\pi^{(k)}}[g(c) = 1 | \pi_1 = \pi^*] - \Pr_{u \leftarrow \mu_{\delta}^k}[u \in \cG_0] - \frac{\epsilon}{6 k}}
  {\Pr_{u \leftarrow \mu_{\delta}^k}[u\in\cG_1-\cG_0]}\Bigr] \nonumber\\ &\qquad\qquad
  - \E_{\pi^*}\Bigl[\frac{
    S_{\pi^*,0} + \Pr_r[\Gamma^*(D(x^*,r))=1]
    (S_{\pi^*,1} - S_{\pi^*,0})}{\Pr_{u \leftarrow \mu_{\delta}^k}[u\in\cG_1-\cG_0]}\Bigr]\label{eq:12}
\end{align}
We bound the second summand of~(\ref{eq:12}).  Consider the set~$\cW$
of puzzles for which both $S_{\pi^*,1}$ and~$S_{\pi^*,0}$ are
not very large.  Formally:
\begin{align}  \label{eq:29}
  \cW := \Bigl\{ \pi \Bigm| \Bigl(S_{\pi,0} \leq
  \bigl(1-\frac{1}{2k}\bigr)\epsilon\Bigr) \land \Bigl(
  S_{\pi,1} \leq \bigl(1-\frac{1}{2k}\bigr)\epsilon\Bigr) \Bigr\}.
\end{align}
Almost surely, $\mu(\cW) \geq 1-\frac{\epsilon}{6k}$: otherwise $\Gen$
would accept one of the sampled puzzles almost surely and recurse.
Thus, we get
\begin{align}
  \E_{\pi^*}[
  &S_{\pi^*,0} + \Pr_r[\Gamma^*(D(x^*,r))=1]
  (S_{\pi^*,1} - S_{\pi^*,0})]\nonumber\\
  &\leq
  \frac{\epsilon}{6k} +
  \E_{\pi^* \leftarrow \cW}[
  S_{\pi^*,0} + \Pr_r[\Gamma^*(D(x^*,r))=1]
  (S_{\pi^*,1} - S_{\pi^*,0})]\\
  &\leq
  \frac{\epsilon}{6k} +
  \E_{\pi^*\leftarrow\cW}[
  S_{\pi^*,0} + \Pr_r[\Gamma^*(D(x^*,r))=1]
  ((1-\frac{1}{2k})\epsilon - S_{\pi^*,0})]\\
  &\leq
  \frac{\epsilon}{6k} +
  \E_{\pi^*\leftarrow\cW}[
  S_{\pi^*,0} + ((1-\frac{1}{2k})\epsilon - S_{\pi^*,0})]\\
  &= \Bigl(1-\frac{1}{3k}\Bigr)\epsilon. \label{eq:33}
\end{align}
We insert into (\ref{eq:12}) (and then use $\Pr[g(u)=1]= \Pr[u \in \cG_0]
+ \delta \Pr[u\in\cG_1 - \cG_0]$) to get
\begin{align}
  \Pr_{\pi,r}[\Gamma(&D(x),r)= 1]\\
  &\geq
  \E_{\pi^*}\Bigl[
  \frac{\Pr_{\pi^{(k)}}[g(c) = 1 | \pi_1 = \pi^*] - \Pr_{u \leftarrow \mu_{\delta}^k}[u \in \cG_0] - \frac{\epsilon}{6 k}}
  {\Pr_{u \leftarrow \mu_{\delta}^k}[u\in\cG_1-\cG_0]}
  -
  \frac{(1-\frac{1}{3k})\epsilon}{\Pr_{u \leftarrow \mu_{\delta}^k}[u\in\cG_1-\cG_0]}\Bigr]\\
  &\geq
  \E_{\pi^*}\Bigl[
  \frac{\Pr_{u\leftarrow\mu_\delta^k}[g(u) = 1] + \epsilon - \Pr_{u \leftarrow \mu_{\delta}^k}[u \in \cG_0] - (1-\frac{1}{6k})\epsilon}
  {\Pr_{u \leftarrow \mu_{\delta}^k}[u\in\cG_1-\cG_0]}\Bigr] \\
  &=
  \frac{\delta \Pr_{u \leftarrow \mu_{\delta}^k}[u\in \cG_1-\cG_0] + \frac{\epsilon}{6k}}
  {\Pr_{u \leftarrow \mu_{\delta}^k}[u\in\cG_1-\cG_0]}
  \geq \delta + \frac{\epsilon}{6k}.
\end{align}
This concludes the proof of Theorem~\ref{'theorem wvp'}.\hfill$\Box$
\fi
\section{Example: Bit Commitment}

\ifshort Theorems~\ref{'theorem bit protocols'} and~\ref{'theorem
  wvp'} can be used to show how to strengthen bit commitment
protocols, which was the main open problem in \cite{HalRab08}.  We
explain this as an example here.  Assume we have given a weak bit
protocol, where a cheating receiver can guess a bit after the
commitment phase with probability $1 - \frac{\beta}{2}$, and a
cheating sender can change the bit he committed to with probability
$\alpha$.  We show that such a protocol can be strengthened if $\alpha
< \beta - 1/\poly(n)$.

We should point out that a different way to prove a similar theorem
exists: one can first show that such a weak bit-commitment protocol
implies one-way functions (using the techniques of \cite{ImpLub89}).
The long sequence of works \cite{HILL99, Naor91, Rompel90, NgOnVa06,
  HaiRei07} imply that one-way functions are sufficient to build bit
commitment protocols (the first two papers will yield statistically
binding protocols, the last three statistically hiding protocols).
However, this will be less efficient and also seems less natural than
the method we use here.

In the appendix, we first define weak bit commitment protocols.  We
then recall a theorem by Valiant \cite{Valian84}, and then show how to
use it to strengthen bit commitment.

Due to space, the details are in the appendix.  However, the result is
a direct composition of the theorems above and known results.

\fi

\iffull
Theorems~\ref{'theorem bit protocols'} and~\ref{'theorem wvp'} can be
used to show how to strengthen bit commitment protocols.  We explain
this as an example here.  Assume we have given a weak bit protocol,
where a cheating receiver can guess a bit after the commitment phase
with probability $1 - \frac{\beta}{2}$, and a cheating sender can
change the bit he committed to with probability $\alpha$.  We show that
such a protocol can be strengthened if $\alpha < \beta - 1/\poly(n)$.

We should point out that a different way to prove a similar theorem
exists: one can first show that such a weak bit-commitment protocol
implies one-way functions (using the techniques of \cite{ImpLub89}).
The long sequence of works \cite{HILL99, Naor91, Rompel90, NgOnVa06,
  HaiRei07} imply that one-way functions are sufficient to build bit
commitment protocols (the first two papers will yield statistically
binding protocols, the last three statistically hiding protocols).
However, this will be less efficient and also seems less natural than
the method we use here.

In the following, we first define weak bit commitment protocols.  We
then recall a Theorem by Valiant \cite{Valian84}, and then show how to
use it to strengthen bit commitment.

\subsection{Weak Bit Commitment Protocols}
We formalize a ``weak'' bit commitment protocol between a sender and a
receiver by considering algorithms~$S(b, r_S)$ and~$R(r_R)$, where~$b$
is the bit which the sender commits to, and $r_S$ and $r_R$ are the
randomness of the sender and receiver respectively.  We denote by
$\Gamma(S(b,r_S) \leftrightarrow R(r_B))$ the communication which one
obtains by running $S(b, r_S)$ interacting with~$R(r_R)$. Also,
$\<S(b, r_S) \leftrightarrow R(r_B)\>_S$ denotes the output which $S$
produces in such an interaction, which for an honest $S$ will be used
later to verify the commitment.  Let $\<S(b, r_S) \leftrightarrow
R(r_B)\>_R$ denote the output receiver $R$ produces which can be
thought of as a guess of $b$.


\begin{definition}
  An $\alpha$-binding $\beta$-hiding bit commitment protocol consists
  of two randomized interactive TM $S(b,r_S)$ and $R(r_R)$, as well as
  a check-algorithm~$R_C$, with the following properties.
  \begin{description}
  \item [Correctness] The protocol works if both parties are honest.
    More concretely, for~$\gamma = \Gamma(S(b,r_S) \leftrightarrow
    R(r_R))$ and $\tau = \<S(b,r_S) \leftrightarrow R(r_R)\>_S$ we have
    that $R_C(b, \gamma, \tau) = 1$ with probability $1-\negl(n)$.
  \item [Binding] A malicious sender cannot open the commitment in two
    ways: For any randomized polynomial time machine $S^*(r_S)$, setting $\gamma :=
    \Gamma(S^*(r_S) \leftrightarrow R(r_R))$, the probability that $S^*$
    outputs $\tau_{0}$ and~$\tau_{1}$ such that $R_C(0,\gamma,\tau_0)
    = 1$ and $R_C(1,\gamma,\tau_1) = 1$ is at most $\alpha$.
  \item[Hiding] For any randomized polynomial time machine~$R^*$,
    $\Pr\bigl[\<S(b,r_S) \leftrightarrow R^*(r_R)\>_R = b\bigr] \leq
    1-\frac{\beta}{2}$, if~$b$ is chosen uniformly at random.
  \end{description}
  If a protocol is $1/p(n)$-binding and $1-1/p(n)$ hiding for all
  polynomial~$p(\cdot)$ and all but finitely many~$n$ we say that it
  is a \emph{strong} bit commitment protocol.
\end{definition}

We point out that our notation is chosen such that for a strong bit
commitment scheme, $\alpha \to 0$ and~$\beta \to 1$.  Given an
$\alpha$-binding $\beta$-hiding bit commitment protocol, we would like
to use it to get a strong bit commitment protocol.  By a simulation
technique \cite{DaKiSa99} this is impossible if~$\alpha \geq \beta$
(there is a simple protocol which achieves this bound for semi-honest
parties without any assumption: with probability~$1-\alpha$ the sender
sends his output bit to the receiver, and otherwise neither party
sends anything).  Our results will show that if~$\alpha < \beta -
1/\poly(n)$ then such a strengthening exists.  Previously, such a
result was only known for $\alpha < \beta - 1/\polylog(n)$
\cite{HalRab08} (if one is restricted to reductions in which the
parties can only use the given protocol interactively, and not to
build a one-way function).

\subsection{Monotone Threshold Functions}

Given a weak protocol~$(S,R)$, we will transform it as follows: the
parties will execute~$(S,R)$ sequentially $k$ times, where the sender
uses random bits as input.  Then, they will apply an ``extraction
protocol'', which is made with the following two properties in mind:
a party who knows at least~$1-\alpha$ fraction of the committed bits
will know the output bit almost surely; a party who has no
information about $1-\beta$ fraction of the input bits will have no
information about the output bit almost surely.  It turns out that
such an extraction process can be modeled as a monotone boolean
circuit, where every wire is used in at most one gate (i.e., read-once
formulas).

To get such a circuit, we use the following lemma.  It can be obtained
by the techniques of Valiant \cite{Valian84}.  Also, it appears in a
more disguised form as Lemma 7 in \cite{DaKiSa99} (where it is used
for the same task we use it here, but not stated in this language).
\begin{lemma}[\cite{Valian84,DaKiSa99}]\label{lem:dks99}
  Let $\alpha$, $\beta$ with~$\alpha < \beta-1/\poly(n)$ be
  efficiently computable.

  There exists a $k\in\poly(n)$ and an efficiently computable monotone
  circuit~$g(m_1,\ldots,m_k)$ where every wire is used in at most one
  gate and such that
  \begin{align}\label{eq:8}
    \Pr[g(\mu_{\beta}^k) = 1] > 1-2^{-n}
  \end{align}
  and
  \begin{align}\label{eq:9}
    \Pr[g(\mu_{\alpha}^k) = 1] < 2^{-n}
  \end{align}
\end{lemma}

\subsection{Strengthening Bit Commitment}
We come to our result of this section.

\begin{theorem}
  Let~$(S,R)$ be an $\alpha$-binding and~$\beta$-hiding bit commitment
  protocol for polynomial time computable functions~$\alpha$ and
  $\beta$ with $\alpha < \beta - 1/\poly(n)$.  Then, there is an
  oblivious black-box construction of a bit commitment scheme
  $(S_0^S, R_0^R)$.
\end{theorem}
\begin{proof}
  Let~$g$ be as guaranteed by Lemma~\ref{lem:dks99} for these
  parameters~$\alpha,\beta$, and~$k$ the input length of~$g$.  The
  players run~$k$ instances of~$(S,R)$ sequentially, where the sender
  commits to a uniform random bit $c_i$ in instance~$i$.  We associate
  each $c_i$ to one of the input wires.  The sender then runs the
  following ``extraction protocol'', in which he uses additional
  variables\footnote{We assume fan-in $2$ on all
    gates.}~$c_{k+1},\ldots,c_{2k-1}$.  We associate those with the
  other wires  in $g$.\footnote{It is advisable to think of
    $g$ as evaluating which values in the following protocols look
    completely random: a $1$ on wire~$i$ signalizes that $c_i$ looks
    random to some party.}  The sender then traverses $g$ as if he were
  evaluating the circuit.  When encountering a gate with input
  wires~$i$, $j$, and output wire~$\ell$, he distinguish two cases.
  If the gate is an OR gate, set~$c_\ell = c_i \oplus c_j$.  If the
  gate is an AND gate, the sender sets~$c_\ell$ to be a completely new
  random value and sends $c_\ell \oplus c_i$ and $c_\ell \oplus c_j$
  to the receiver.  Once the sender ``evaluated'' $g$ in this way, he
  sends $b \oplus c_{2k-1}$ to the receiver (where~$b$ is the input
  to the sender, and~$c_{2k-1}$ is the bit associated with the output
  wire of $g$).

  To open the commitment, the sender sends all the opening information
  for the individual positions to the receiver.  The receiver then
  checks if the extraction phase was done consistently, and accepts if
  all these tests succeed and the output matches.

  \emph{Hiding:} We would like to use Theorem~\ref{'theorem bit
    protocols'}.  For this, it only remains to argue that the
  extraction protocol is information theoretically secure.  For any
  $\beta$-hiding random variables, we define a random variable $H$
  over~$\{0,1\}$ by fixing $\Pr[H = 1 |X = x, Z = z] =
  \frac{\min(\Pr[X=0, Z=z], \Pr[X=1,Z=z])}{\Pr[X=x,Z=z]}$.  One checks
  that for any function $f: \cZ \to \{0,1\}$ we have $\Pr[f(Z) = X | H
  = 1] = \frac{1}{2}$ and $\Pr[H = 1] = 1 - \frac{\beta}{2}$ (the
  point of $H$ is that it is $1$ exactly if $Z$ gives no information
  about $X$, and furthermore $H$ is often $1$).  We get random
  variables $H_1, \ldots, H_k$ in this way, and evaluate the circuit
  $g(H_1,\ldots,H_k)$.  One sees per induction that $Z_1,\ldots,Z_k$
  together with the communication produced gives no information about
  the bit corresponding to a wire iff the corresponding value when
  evaluating $g(H_1,\ldots,H_k)$ is one.  Since the probability that
  the output is $1$ is $1-2^{-n}$, we get the information theoretic
  security.

  \emph{Binding:} We can interpret the bit commitment protocol as an
  interactive weakly verifiable puzzle: in the interaction, the
  receiver is the person posing the puzzle, and the sender is the
  solver.  In order to solve the puzzle, the sender needs to send two
  valid openings to the receiver.

  In order to break the resulting puzzle, the sender needs to solve
  the subpuzzles in all positions $a_i$ for some input for which
  $g(a_1,\ldots,a_k) = 1$.  Using Theorem~\ref{'theorem wvp'} for
  $\delta = \beta$ thus gives the result.
\end{proof}

\fi
\iffull
\section{Acknowledgments}
We would like to thank the anonymous referees for useful comments.
\fi
{
\footnotesize
\bibliographystyle{alpha}
\bibliography{db}
}

\ifshort

\appendix

\section{Proof of Theorem~\ref{'theorem single pred'}}

\begin{proof}
  We describe algorithm $\Gen$.  First, obtain an estimate
  \begin{align}
    \Delta :\approx \Pr_{r,x}[C(x,P(x), r)=1] -
    \Pr_{r,x}[C(x,1-P(x), r)=1]
  \end{align}
  such that almost surely $\Delta$ is within~$\epsilon/4$ of the
  actual quantity.  If~$|\Delta| < 3\epsilon/4$, we can return~$\delta
  = 1$, $S = \{0,1\}^n$, and a circuit~$Q$ which guesses a uniform
  random bit.  If~$\Delta < -3\epsilon/4$ replace $C$ with the circuit
  which outputs~$1-C$ in the following argument.  Thus, from now on
  assume~$\Delta > 3\epsilon/4$ and that the actual quantity is at
  least~$\epsilon/2$.

  Sample random strings~$r_1,\ldots,r_m$ for $C$, where~$m =
  100n/\epsilon^2$, and let~$C'(x,b,i)$ be the circuit which computes
  $C(x,b,r_i)$.  Using a Chernoff bound, we see that for all $x \in \{0,1\}^n$
  \begin{align}
    \Pr_{r}[C&(x,P(x), r)=1] -
    \Pr_r[C(x,1-P(x), r)]=1]= \nonumber\\
    & \Pr_{i \in [m]}[C'(x,P(x), i)=1] - \Pr_{i \in [m]}[C'(x,1-P(x), i)]=1]
    \pm \epsilon/4
  \end{align}
  almost surely.

  Define, for any~$x$,
  \begin{align}\label{eq:4}
    \Delta_x := \Pr_{i \in [m]}[C'(x,P(x), r_i)=1] - \Pr_{i \in [m]}[C'(x,1-P(x), r_i)=1].
  \end{align}
  Because we define~$\Delta_x$ using~$C'$ instead of $C$, we can
  compute~$\Delta_x$ exactly for a given~$x$.  Now, order the~$x$
  according to $\Delta_x$: let~$x_1 \preceq x_2$ if $\Delta_{x_1} <
  \Delta_{x_2}$, or both~$\Delta_{x_1} = \Delta_{x_2}$ and~$x_1 \leq_L
  x_2$, where~$\leq_L$ is the lexicographic ordering on bitstrings.
  We can compute~$x_1 \preceq x_2$ efficiently given $(x_1,P(x_1))$
  and $(x_2, P(x_2))$.

  We claim that we can find~$x^*$ such that almost surely (we
  assume~$\epsilon > 10\cdot 2^{-n}$, otherwise we can get the theorem
  with exhaustive search)
  \begin{align}\label{eq:3}
    \frac{\epsilon}{20} \leq \frac{1}{2^n}
    \sum_{x\preceq x^*} \Delta_x \leq \frac{\epsilon}{10}.
  \end{align}
  We pick~$50 n/\epsilon$ candidates, then almost surely one of them
  satisfies~(\ref{eq:3}) with a safety margin of~$\epsilon/50$.  For
  each of those candidates we estimate $\frac{1}{2^n} \sum_{x\preceq
    x^*} \Delta_x$ up to an error of~$\epsilon/100$, and keep one for
  which almost surely~(\ref{eq:3}) is satisfied.  We let $S(x,P(x))$
  be the circuit which recognizes the set~$S^* := \{x | x \preceq
  x^*\}$, estimate~$\delta' := |S^*|/2^n$ almost surely within an
  error of~$\epsilon/1000$, and output~$\delta := \delta' -
  \epsilon/1000$.  The situation at this moment is illustrated in
  Figure~\ref{fig:littlepicture}, and it is clear that the properties
  ``large set'' and ``indistinguishability'' are satisfied.

  We next describe~$Q$.  On input $x$, $Q$ calculates (exactly)
  \begin{align}
    \Pr_{i \in [m]}[C'(x,1, i) = 1] - \Pr_{i \in [m]}[C'(x,0, i) = 1] = (2P(x)-1)\Delta_x\;.
  \end{align}
  If $(2P(x) - 1)\Delta_x \geq \Delta_{x^*}$ (where~$\Delta_{x^*}$ is
  defined by~(\ref{eq:4}) for the element~$x^*$ which defines~$S$),
  then output~$1$, if $(2P(x) - 1)\Delta_x \leq - \Delta_{x^*}$ output
  $0$.  If neither of the previous cases apply, output~$1$ with
  probability $\frac{1}{2}(1+\frac{(2P(x) -
    1)\Delta_x}{\Delta_{x^*}})$.

  To analyze the success probability of $Q$, we distinguish two cases.
  If~$x\notin S$, we know that~$\Delta_{x} \geq \Delta_{x^*}$.
  Therefore, in this case, we get the correct answer with probability
  $1$.  If~$x \in S$, it is also easy to check that this will give the
  correct answer with probability
  $\max\{\frac{1}{2}(1+\frac{\Delta_x}{\Delta_{x^*}}), 0\}$, and thus,
  on average $\frac{1}{|S|}\sum_{x \in S}
  \max\{\frac{1}{2}(1+\frac{\Delta_x}{\Delta_{x^*}}), 0\} \geq
  \frac{1}{|S|}\sum_{x \in S}
  \frac{1}{2}(1+\frac{\Delta_x}{\Delta_{x^*}}) \geq
  \frac{\epsilon}{20}$, using~(\ref{eq:3}).  In total, we have
  probability at least $\mu(S)(\frac{1}{2}+\frac{\epsilon}{40}) +
  (1-\mu(S))$ of answering correctly.  Since $\mu(S) \geq \delta$,
  this quantity is at least $1 - \frac{\delta}{2}$, which implies
  ``predictability''.  It is possible to make $Q$ deterministic by
  trying all possible values for the randomness and estimating the
  probability of it being correct.

  In order to get the additional property 1, we first run the
  above algorithm with input $\epsilon/3$ instead of $\epsilon$.  If
  $\delta > 1 - 2\epsilon/3$, we instead output the set containing all
  elements and return $1$ in place of $\delta$.  Note that
  indistinguishability still holds because we only add a fraction of
  $2\epsilon/3$ elements to $S$.  If $\delta \leq 1 - 2\epsilon/3$, we
  enlarge $S$ by at least $\epsilon/2$ and at most $2\epsilon/3$; this
  can be done by finding a new candidate for $x^*$ as above.  We then
  output the new set and $\delta' := \delta + \frac{\epsilon}{2}$.

  The additional properties 2, 3 and 4 follow by inspection of the proof.
\end{proof}

\section{Proof of Theorem~\ref{'theorem many pred'}}
\begin{proof}
  For any fixed tuple $t_i = (x_1,b_1,\ldots,x_{i-1},b_{i-1})$,
  consider the circuit~$C_{t_i}(x,b,r)$ which uses $r$ to pick
  random~$x_j$ for~$j > i$, and runs~$C^{(k)}(t_i,x,b, x_{i+1},
  P(x_{i+1}),\ldots,x_{k},P(x_k))$.\footnote{Formally, $C_{t_i}$ may
    not be a small circuit because at this point we do not assume $P$
    to be efficiently samplable, and $C_{t_i}$ seems to need to use
    $r$ to sample pairs $(x_j, P(x_j))$ for $j > i$.  However, we can
    think of $C_{t_i}$ as oracle circuit with oracle access to $P$ at
    this moment.  Inspection of the previous proof shows that later we
    can remove the calls to $P$, as the $x_{j}$ with $j > i$ can be
    fixed.}  We let $\GenS$ be the algorithm which invokes $\Gen$ with
  parameter $\frac{\epsilon}{4k}$ from Theorem~\ref{'theorem single
    pred'} on the circuit~$C_{t_i}$ and then returns the circuit
  recognizing a set from there.

  We next describe $\Gen$: For~$\ell = n k/\epsilon$ iterations, pick
  a random~$i \in \{0, \ldots, k-1\}$, use the procedure in
  Experiment~2 until loop $i$, and run algorithm $\Gen$ from
  Theorem~\ref{'theorem single pred'} with parameter
  $\frac{\epsilon}{4k}$.  This yields a parameter $\delta$ and a
  circuit $Q$.  We output the pair~$(Q,\delta)$ for the
  smallest~$\delta$ ever encountered.  Since $k$ and $\epsilon$ are
  polynomial in $n$, almost surely every time Theorem~\ref{'theorem
    single pred'} is used the almost surely part happens.  Thus, we
  get the property ``predictability'' (and in fact the stronger
  property listed under additionally) immediately.  We now argue
  ``large sets'': consider the random variable $\delta$ when we pick a
  random $i$, simulate an execution up to iteration $i$ of
  Experiment~2, then run $\Gen$ from Theorem~\ref{'theorem single
    pred'}.  Let $\delta^*$ be the $\frac{\epsilon}{k}$-quantile of
  this distribution, i.e., the smallest value such that with
  probability $\frac{\epsilon}{k}$ the value of $\delta$ is at most
  $\delta^*$.  The probability that a value not bigger than $\delta^*$
  is output by $\Gen$ is at least $1-(1-\frac{\epsilon}{k})^\ell > 1 -
  2^{-n}$, in which case ``large sets'' is satisfied.

  We show ``indistinguishability'' with a standard hybrid argument.
  Consider the Experiment~$H_j$:
\begin{lstlisting}
*\textbf{Random Experiment $H_j$:}*
        for $i := 1$ to $k$ do
           $t_i := (x_1,b_1,\ldots,t_{i-1}, b_{i-1})$
           $S_{i}^* := \GenS(t_i)$
           $x_i \leftarrow \{0,1\}^n$
           if $i \leq j$ and $x_i \in S^*_t$ then
              $b_i \leftarrow \{0,1\}$
           else
              $b_i := P(x_i)$
           end if
        end for
        $r \leftarrow \{0,1\}^*$
        output $C^{(k)}(x_1,b_1,\ldots,x_k,b_k,r)$
\end{lstlisting}

Experiment~$H_0$ is equivalent to Experiment~$1$, Experiment $H_{k}$
is the same as Experiment~$2$.  Applying Theorem~\ref{'theorem single
  pred'} we get that for every fixed~$x_1,\ldots,x_{j-1},
b_1,\ldots,b_{j-1}$, almost surely
\begin{align}
  \Bigl|
  \Pr_{x_i,\ldots,x_k}[&C^{(k)}(x_1,b_1,\ldots,x_{j-1},b_{j-1},x_j,b_j^{(j-1)},\ldots,x_k,P(x_k)) = 1] -\nonumber\\
  \Pr_{x_i,\ldots,x_k}[&C^{(k)}(x_1,b_1,\ldots,x_{j-1},b_{j-1},x_j,b_j^{(j)},\ldots,x_k,P(x_k)) = 1] \Bigr|
  \leq \epsilon/4k\;,
\end{align}
where $b_j^{(j-1)}$ is chosen as $b_j^{(j-1)} = P(x_j)$ in experiment
$H_{j-1}$, and $b_{j}^{(j)}$ is chosen the same way as $b_j$ is chosen
in experiment $H_{j}$ (in Theorem~\ref{'theorem single pred'} the bit
is flipped, but when using a uniform bit instead of flipping it the
distinguishing probability only gets smaller).  Applying the triangle
inequality $k-1$ times we get that almost surely the difference of the
probabilities in Experiment~1 and Experiment~2 is at most
$\frac{\epsilon}{2}$.  Since ``almost surely'' means with
probabilities $1-2^{-n}\poly(n) > 1-\frac{\epsilon}{2}$, we get
``indistinguishability''.

We already showed the additional Property 1.  Properties 2,3, and 4
follow by inspection.
\end{proof}

\section{Finishing the Proof of Theorem \ref{'theorem wvp'}}
The algorithm itself is described in the main part of the paper.  Here
we show that it achieves the guarantees as promised.
\paragraph{Overview of Correctness}
The interesting case is when $\Gen$ does not recurse.  In this case we
know that $C$ has higher success probability than $\Pr_{u\leftarrow
  \mu_\delta^k}[g(u)=1]$, but for most $\pi^*$, the surpluses
$S_{\pi^*,0}$ and $S_{\pi^*,1}$ are less than
$(1-\frac{1}{k})\epsilon$.  Intuitively, then $C$ is correct on the
first coordinate unusually often when $c \in \cG_1 - \cG_0$ (as this
is the only time that being correct on the first coordinate helps).
If we could assume that 1) that the algorithm \emph{always} outputs an
answer, and 2) for \emph{every} $\pi^*$, the surpluses, $S_{\pi^*,0}$
and $S_{\pi^*,1}$ are less than $(1-\frac{1}{k})\epsilon$, then the
theorem would follow by straight-forward manipulations of probability.

Unfortunately these assumptions are not true, but the proof below
shows that because these assumptions only fail slightly, not much is
lost.  Informally, Equations~\ref{eq:22}-\ref{eq:27} show that if the
algorithm fails to output an answer it is either because
$\Pr_{\pi^{(k)}}[c \in \cG_1 - \cG_0|\pi_1 = \pi^*]$ is very small (in
which case this $\pi^*$ will not contribute much anyhow), or because
we are unlucky (which happens with very small probability).
Additionally, Equations~\ref{eq:29}-\ref{eq:33} show that because we
did not find a $\pi^*$ with large surplus, we can assume that (unless
we were very unlucky) there are few $\pi^*$ with large surpluses,
which cannot have undue influence.

\paragraph{Analysis of Correctness}

Consider first the case that we find $(x^*, b)$ for which
$\Stilde_{x^*,b} \geq (1-\frac{3}{4k})\epsilon$.  We can assume that
$S_{x^*,b} \geq (1-\frac{1}{k})\epsilon$, since the error is at
most~$\epsilon/(4k)$ almost surely.  Thus, we satisfy all the
requirements to use $\Gen$ with~$k-1$ (using $x^*$ as the first input
and $g(b, \cdot)$ as the monotone function with $k$ inputs), which
will return a non-rewinding circuit for which
$\Pr_{(x,\Gamma),r}[\Gamma(D(x,r)) = 1] \geq \delta +
(1-\frac{1}{k})\epsilon/6(k-1) = \delta + \epsilon/6k$.  The remaining
properties are easily verified.

The more interesting case is if $\Gen$ does not recurse.  First, we get,
for any puzzle~$\pi^* = (x^*,\Gamma^*)$ (simply using~(\ref{eq:13})
and~$\cG_0 \subseteq \cG_1$):
\begin{align}
  \Pr_{u \leftarrow \mu_{\delta}^k}[u\in\cG_1 - \cG_0] = \Pr_{\pi^{(k)}}[c \in \cG_1 - \cG_0|\pi_1 = \pi^*] -
  (S_{\pi^*,1} - S_{\pi^*,0})
\end{align}
and thus, still fixing~$\pi^*$ and multiplying by  $\Pr_r[\Gamma^*(D(x^*,r))=1] / \Pr_{u \leftarrow \mu_{\delta}^k}[u\in\cG_1 - \cG_0]$:
\begin{align}
  \Pr_r[\Gamma^*(D(x^*,r))=1] &=
  \frac{\Pr_r[\Gamma^*(D(x^*,r))=1] \Pr_{\pi^{(k)}}[c \in \cG_1 - \cG_0|\pi_1 = \pi^*]}
  {\Pr_{u \leftarrow \mu_{\delta}^k}[u\in\cG_1-\cG_0]}  \nonumber\\
  &\qquad\qquad-
  \frac{\Pr_r[\Gamma^*(D(x^*,r))=1] (S_{\pi^*,1}-S_{\pi^*,0})}
  {\Pr_{u \leftarrow \mu_{\delta}^k}[u\in\cG_1-\cG_0]}.\label{eq:10}
\end{align}
We bound the first summand in (\ref{eq:10}):
\begin{align}  \label{eq:22}
  \Pr_r[\Gamma^*(&D(x^*,r))=1] \Pr_{\pi^{(k)}}[c \in \cG_1 - \cG_0|\pi_1 = \pi^*]\nonumber\\
  &=
  \Pr_r[\Gamma^*(D(x^*,r))\neq \bot]
  \Pr_{\pi^{(k)}}[c_1 = 1 | c \in \cG_1 - \cG_0, \pi_1 = \pi^*]
  \Pr_{\pi^{(k)}}[c \in \cG_1 - \cG_0 | \pi_1 = \pi^*].
\end{align}
If $\Pr[c\in\cG_1 - \cG_0|\pi_1 = \pi^*] \leq \frac{\epsilon}{6k}$,
then~$0 \geq \Pr[c_1 = 1 | c \in \cG_1 - \cG_0, x_1 = \pi^*] \Pr[c \in
\cG_1 - \cG_0 | \pi_1 = \pi^*] -\frac{\epsilon}{6k}$.  If
$\Pr[c\in\cG_1 - \cG_0| \pi_1 = \pi^*] > \frac{\epsilon}{6k}$ then
$\Pr[\Gamma^*(C(x^*))\neq\bot] \geq 1-\frac{\epsilon}{6k}$ since $D$
only outputs~$\bot$ if after $\frac{6k}{\epsilon}\log(6k/\epsilon)$
none of the elements~$c$ was in~$\cG_1-\cG_0$.  In both cases:
\begin{align}
  \Pr_r[\Gamma^*(&D(x^*,r))\neq \bot]
  \Pr_{\pi^{(k)}}[c_1 = 1 | c \in \cG_1 - \cG_0, \pi_1 = \pi^*]
  \Pr_{\pi^{(k)}}[c \in \cG_1 - \cG_0 | \pi_1 = \pi^*]\\
  &\geq
  \Pr_{\pi^{(k)}}[c_1 = 1 | c \in \cG_1 - \cG_0 , \pi_1 = \pi^*]
  \Pr_{\pi^{(k)}}[c \in \cG_1 - \cG_0 | \pi_1 = \pi^*] - \frac{\epsilon}{6 k}\\
  &=
  \Pr_{\pi^{(k)}}[c_1 = 1 \land c \in \cG_1 - \cG_0| \pi_1 = \pi^*] - \frac{\epsilon}{6 k}\\
  &=
  \Pr_{\pi^{(k)}}[g(c) = 1 | \pi_1 = \pi^*] - \Pr_{\pi^{(k)}}[c \in \cG_0 | \pi_1 = \pi^*] -
  \frac{\epsilon}{6 k}\\
  &=\Pr_{\pi^{(k)}}[g(c) = 1 | \pi_1 = \pi^*] - \Pr_{u \leftarrow \mu_{\delta}^k}[u \in \cG_0] - S_{\pi^*,0}
  - \frac{\epsilon}{6 k}  \label{eq:27}
\end{align}
Inserting into (\ref{eq:10}) gives
\begin{align}
  \E_{\pi^*}\bigl[ \Pr_{r}[\Gamma^*(D(x^*,r))= 1]\bigr]
  &\geq
  \E_{\pi^*}\Bigl[
  \frac{\Pr_{\pi^{(k)}}[g(c) = 1 | \pi_1 = \pi^*] - \Pr_{u \leftarrow \mu_{\delta}^k}[u \in \cG_0] - \frac{\epsilon}{6 k}}
  {\Pr_{u \leftarrow \mu_{\delta}^k}[u\in\cG_1-\cG_0]}\Bigr] \nonumber\\ &\qquad\qquad
  - \E_{\pi^*}\Bigl[\frac{
    S_{\pi^*,0} + \Pr_r[\Gamma^*(D(x^*,r))=1]
    (S_{\pi^*,1} - S_{\pi^*,0})}{\Pr_{u \leftarrow \mu_{\delta}^k}[u\in\cG_1-\cG_0]}\Bigr]\label{eq:12}
\end{align}
We bound the second summand of~(\ref{eq:12}).  Consider the set~$\cW$
of puzzles for which both $S_{\pi^*,1}$ and~$S_{\pi^*,0}$ are
not very large.  Formally:
\begin{align}  \label{eq:29}
  \cW := \Bigl\{ \pi \Bigm| \Bigl(S_{\pi,0} \leq
  \bigl(1-\frac{1}{2k}\bigr)\epsilon\Bigr) \land \Bigl(
  S_{\pi,1} \leq \bigl(1-\frac{1}{2k}\bigr)\epsilon\Bigr) \Bigr\}.
\end{align}
Almost surely, $\mu(\cW) \geq 1-\frac{\epsilon}{6k}$: otherwise $\Gen$
would accept one of the sampled puzzles almost surely and recurse.
Thus, we get
\begin{align}
  \E_{\pi^*}[
  &S_{\pi^*,0} + \Pr_r[\Gamma^*(D(x^*,r))=1]
  (S_{\pi^*,1} - S_{\pi^*,0})]\nonumber\\
  &\leq
  \frac{\epsilon}{6k} +
  \E_{\pi^* \leftarrow \cW}[
  S_{\pi^*,0} + \Pr_r[\Gamma^*(D(x^*,r))=1]
  (S_{\pi^*,1} - S_{\pi^*,0})]\\
  &\leq
  \frac{\epsilon}{6k} +
  \E_{\pi^*\leftarrow\cW}[
  S_{\pi^*,0} + \Pr_r[\Gamma^*(D(x^*,r))=1]
  ((1-\frac{1}{2k})\epsilon - S_{\pi^*,0})]\\
  &\leq
  \frac{\epsilon}{6k} +
  \E_{\pi^*\leftarrow\cW}[
  S_{\pi^*,0} + ((1-\frac{1}{2k})\epsilon - S_{\pi^*,0})]\\
  &= \Bigl(1-\frac{1}{3k}\Bigr)\epsilon. \label{eq:33}
\end{align}
We insert into (\ref{eq:12}) (and then use $\Pr[g(u)=1]= \Pr[u \in \cG_0]
+ \delta \Pr[u\in\cG_1 - \cG_0]$) to get
\begin{align}
  \Pr_{\pi,r}[\Gamma(D(x),r)= 1]
  &\geq
  \E_{\pi^*}\Bigl[
  \frac{\Pr_{\pi^{(k)}}[g(c) = 1 | \pi_1 = \pi^*] - \Pr_{u \leftarrow \mu_{\delta}^k}[u \in \cG_0] - \frac{\epsilon}{6 k}}
  {\Pr_{u \leftarrow \mu_{\delta}^k}[u\in\cG_1-\cG_0]}
  -
  \frac{(1-\frac{1}{3k})\epsilon}{\Pr_{u \leftarrow \mu_{\delta}^k}[u\in\cG_1-\cG_0]}\Bigr]\\
  &\geq
  \E_{\pi^*}\Bigl[
  \frac{\Pr_{u\leftarrow\mu_\delta^k}[g(u) = 1] + \epsilon - \Pr_{u \leftarrow \mu_{\delta}^k}[u \in \cG_0] - (1-\frac{1}{6k})\epsilon}
  {\Pr_{u \leftarrow \mu_{\delta}^k}[u\in\cG_1-\cG_0]}\Bigr] \\
  &=
  \frac{\delta \Pr_{u \leftarrow \mu_{\delta}^k}[u\in \cG_1-\cG_0] + \frac{\epsilon}{6k}}
  {\Pr_{u \leftarrow \mu_{\delta}^k}[u\in\cG_1-\cG_0]}
  \geq \delta + \frac{\epsilon}{6k}.
\end{align}
This concludes the proof of Theorem~\ref{'theorem wvp'}.\hfill$\Box$

\section{Example: Bit Commitment}
Theorems~\ref{'theorem bit protocols'} and~\ref{'theorem wvp'} can be
used to show how to strengthen bit commitment protocols.  We explain
this as an example here.  Assume we have given a weak bit protocol,
where a cheating receiver can guess a bit after the commitment phase
with probability $1 - \frac{\beta}{2}$, and a cheating sender can
change the bit he committed to with probability $\alpha$.  We show that
such a protocol can be strengthened if $\alpha < \beta - 1/\poly(n)$.

We should point out that a different way to prove a similar theorem
exists: one can first show that such a weak bit-commitment protocol
implies one-way functions (using the techniques of \cite{ImpLub89}).
The long sequence of works \cite{HILL99, Naor91, Rompel90, NgOnVa06,
  HaiRei07} imply that one-way functions are sufficient to build bit
commitment protocols (the first two papers will yield statistically
binding protocols, the last three statistically hiding protocols).
However, this will be less efficient and also seems less natural than
the method we use here.

In the following, we first define weak bit commitment protocols.  We
then recall a Theorem by Valiant \cite{Valian84}, and then show how to
use it to strengthen bit commitment.

\subsection{Weak Bit Commitment Protocols}
We formalize a ``weak'' bit commitment protocol between a sender and a
receiver by considering algorithms~$S(b, r_S)$ and~$R(r_R)$, where~$b$
is the bit which the sender commits to, and $r_S$ and $r_R$ are the
randomness of the sender and receiver respectively.  We denote by
$\Gamma(S(b,r_S) \leftrightarrow R(r_B))$ the communication which one
obtains by running $S(b, r_S)$ interacting with~$R(r_R)$. Also,
$\<S(b, r_S) \leftrightarrow R(r_B)\>_S$ denotes the output which $S$
produces in such an interaction, which for an honest $S$ will be used
later to verify the commitment.  Let $\<S(b, r_S) \leftrightarrow
R(r_B)\>_R$ denote the output receiver $R$ produces which can be
thought of as a guess of $b$.


\begin{definition}
  An $\alpha$-binding $\beta$-hiding bit commitment protocol consists
  of two randomized interactive TM $S(b,r_S)$ and $R(r_R)$, as well as
  a check-algorithm~$R_C$, with the following properties.
  \begin{description}
  \item [Correctness] The protocol works if both parties are honest.
    More concretely, for~$\gamma = \Gamma(S(b,r_S) \leftrightarrow
    R(r_R))$ and $\tau = \<S(b,r_S) \leftrightarrow R(r_R)\>_S$ we have
    that $R_C(b, \gamma, \tau) = 1$ with probability $1-\negl(n)$.
  \item [Binding] A malicious sender cannot open the commitment in two
    ways: For any randomized polynomial time machine $S^*(r_S)$, setting $\gamma :=
    \Gamma(S^*(r_S) \leftrightarrow R(r_R))$, the probability that $S^*$
    outputs $\tau_{0}$ and~$\tau_{1}$ such that $R_C(0,\gamma,\tau_0)
    = 1$ and $R_C(1,\gamma,\tau_1) = 1$ is at most $\alpha$.
  \item[Hiding] For any randomized polynomial time machine~$R^*$,
    $\Pr\bigl[\<S(b,r_S) \leftrightarrow R^*(r_R)\>_R = b\bigr] \leq
    1-\frac{\beta}{2}$, if~$b$ is chosen uniformly at random.
  \end{description}
  If a protocol is $1/p(n)$-binding and $1-1/p(n)$ hiding for all
  polynomial~$p(\cdot)$ and all but finitely many~$n$ we say that it
  is a \emph{strong} bit commitment protocol.
\end{definition}

We point out that our notation is chosen such that for a strong bit
commitment scheme, $\alpha \to 0$ and~$\beta \to 1$.  Given an
$\alpha$-binding $\beta$-hiding bit commitment protocol, we would like
to use it to get a strong bit commitment protocol.  By a simulation
technique \cite{DaKiSa99} this is impossible if~$\alpha \geq \beta$
(there is a simple protocol which achieves this bound for semi-honest
parties without any assumption: with probability~$1-\alpha$ the sender
sends his output bit to the receiver, and otherwise neither party
sends anything).  Our results will show that if~$\alpha < \beta -
1/\poly(n)$ then such a strengthening exists.  Previously, such a
result was only known for $\alpha < \beta - 1/\polylog(n)$
\cite{HalRab08} (if one is restricted to reductions in which the
parties can only use the given protocol interactively, and not to
build a one-way function).

\subsection{Monotone Threshold Functions}

Given a weak protocol~$(S,R)$, we will transform it as follows: the
parties will execute~$(S,R)$ sequentially $k$ times, where the sender
uses random bits as input.  Then, they will apply an ``extraction
protocol'', which is made with the following two properties in mind:
a party who knows at least~$1-\alpha$ fraction of the committed bits
will know the output bit almost surely; a party who has no
information about $1-\beta$ fraction of the input bits will have no
information about the output bit almost surely.  It turns out that
such an extraction process can be modeled as a monotone boolean
circuit, where every wire is used in at most one gate (i.e., read-once
formulas).

To get such a circuit, we use the following lemma.  It can be obtained
by the techniques of Valiant \cite{Valian84}.  Also, it appears in a
more disguised form as Lemma 7 in \cite{DaKiSa99} (where it is used
for the same task we use it here, but not stated in this language).
\begin{lemma}[\cite{Valian84,DaKiSa99}]\label{lem:dks99}
  Let $\alpha$, $\beta$ with~$\alpha < \beta-1/\poly(n)$ be
  efficiently computable.

  There exists a $k\in\poly(n)$ and an efficiently computable monotone
  circuit~$g(m_1,\ldots,m_k)$ where every wire is used in at most one
  gate and such that
  \begin{align}\label{eq:8}
    \Pr[g(\mu_{\beta}^k) = 1] > 1-2^{-n}
  \end{align}
  and
  \begin{align}\label{eq:9}
    \Pr[g(\mu_{\alpha}^k) = 1] < 2^{-n}
  \end{align}
\end{lemma}

\subsection{Strengthening Bit Commitment}
We come to our result of this section.

\begin{theorem}
  Let~$(S,R)$ be an $\alpha$-binding and~$\beta$-hiding bit commitment
  protocol for polynomial time computable functions~$\alpha$ and
  $\beta$ with $\alpha < \beta - 1/\poly(n)$.  Then, there is an
  oblivious black-box construction of a bit commitment scheme
  $(S_0^S, R_0^R)$.
\end{theorem}
\begin{proof}
  Let~$g$ be as guaranteed by Lemma~\ref{lem:dks99} for these
  parameters~$\alpha,\beta$, and~$k$ the input length of~$g$.  The
  players run~$k$ instances of~$(S,R)$ sequentially, where the sender
  commits to a uniform random bit $c_i$ in instance~$i$.  We associate
  each $c_i$ to one of the input wires.  The sender then runs the
  following ``extraction protocol'', in which he uses additional
  variables\footnote{We assume fan-in $2$ on all
    gates.}~$c_{k+1},\ldots,c_{2k-1}$.  We associate those with the
  other wires  in $g$.\footnote{It is advisable to think of
    $g$ as evaluating which values in the following protocols look
    completely random: a $1$ on wire~$i$ signalizes that $c_i$ looks
    random to some party.}  The sender then traverses $g$ as if he were
  evaluating the circuit.  When encountering a gate with input
  wires~$i$, $j$, and output wire~$\ell$, he distinguish two cases.
  If the gate is an OR gate, set~$c_\ell = c_i \oplus c_j$.  If the
  gate is an AND gate, the sender sets~$c_\ell$ to be a completely new
  random value and sends $c_\ell \oplus c_i$ and $c_\ell \oplus c_j$
  to the receiver.  Once the sender ``evaluated'' $g$ in this way, he
  sends $b \oplus c_{2k-1}$ to the receiver (where~$b$ is the input
  to the sender, and~$c_{2k-1}$ is the bit associated with the output
  wire of $g$).

  To open the commitment, the sender sends all the opening information
  for the individual positions to the receiver.  The receiver then
  checks if the extraction phase was done consistently, and accepts if
  all these tests succeed and the output matches.

  \emph{Hiding:} We would like to use Theorem~\ref{'theorem bit
    protocols'}.  For this, it only remains to argue that the
  extraction protocol is information theoretically secure.  For any
  $\beta$-hiding random variables, we define a random variable $H$
  over~$\{0,1\}$ by fixing $\Pr[H = 1 |X = x, Z = z] =
  \frac{\min(\Pr[X=0, Z=z], \Pr[X=1,Z=z])}{\Pr[X=x,Z=z]}$.  One checks
  that for any function $f: \cZ \to \{0,1\}$ we have $\Pr[f(Z) = X | H
  = 1] = \frac{1}{2}$ and $\Pr[H = 1] = 1 - \frac{\beta}{2}$ (the
  point of $H$ is that it is $1$ exactly if $Z$ gives no information
  about $X$, and furthermore $H$ is often $1$).  We get random
  variables $H_1, \ldots, H_k$ in this way, and evaluate the circuit
  $g(H_1,\ldots,H_k)$.  One sees per induction that $Z_1,\ldots,Z_k$
  together with the communication produced gives no information about
  the bit corresponding to a wire iff the corresponding value when
  evaluating $g(H_1,\ldots,H_k)$ is one.  Since the probability that
  the output is $1$ is $1-2^{-n}$, we get the information theoretic
  security.

  \emph{Binding:} We can interpret the bit commitment protocol as an
  interactive weakly verifiable puzzle: in the interaction, the
  receiver is the person posing the puzzle, and the sender is the
  solver.  In order to solve the puzzle, the sender needs to send two
  valid openings to the receiver.

  In order to break the resulting puzzle, the sender needs to solve
  the subpuzzles in all positions $a_i$ for some input for which
  $g(a_1,\ldots,a_k) = 1$.  Using Theorem~\ref{'theorem wvp'} for
  $\delta = \beta$ thus gives the result.
\end{proof}


\fi
\end{document}